%% file: main.tex
\documentclass[sigconf]{acmart}
\usepackage{amsmath}
\usepackage{mathtools}
\usepackage{amsthm}
\usepackage{caption}
\usepackage{microtype}
\usepackage{graphicx}
\usepackage{subfigure}
\usepackage{booktabs}
\usepackage{hyperref}
\usepackage{enumitem}
\usepackage{multirow}
\usepackage{multicol}
\usepackage{ulem}
\usepackage{algorithm}
\usepackage{algorithmic}
\usepackage{tabularx}
\usepackage{subcaption}

\newtheorem{theorem}{Theorem}
\usepackage{placeins}
\usepackage[table]{xcolor}
\usepackage{amsmath} 
\usepackage{bm}     

\AtBeginDocument{%
  }

\setcopyright{acmlicensed}
\copyrightyear{2018}
\acmYear{2018}
\acmDOI{XXXXXXX.XXXXXXX}

\acmConference[Conference acronym 'XX]{Make sure to enter the correct
  conference title from your rights confirmation emai}{June 03--05,
  2018}{Woodstock, NY}
\acmISBN{978-1-4503-XXXX-X/18/06}

\begin{document}

\title{An Exterior-Embedding Neural Operator Framework for Preserving Conservation Laws}
\def\method{method}

\author{Huanshuo Dong}
\authornote{Both authors contributed equally to this research.}
\email{bingo000@mail.ustc.edu.cn}
\affiliation{%
  \institution{University of Science and Technology of China}
  \city{Anhui}
  \state{Hefei}
  \country{China}
}

\author{Hong Wang}
\authornote{Both authors contributed equally to this research.}
\affiliation{%
  \institution{University of Science and Technology of China}
  \city{Anhui}
  \state{Hefei}
  \country{China}
}

\author{Hao Wu}
\authornote{Both authors contributed equally to this research.}
\affiliation{%
  \institution{Tsinghua University}
  \city{Beijing}
  \state{Haidian District}
  \country{China}
}

\author{Zhiwei Zhuang}
\affiliation{%
  \institution{University of Science and Technology of China}
  \city{Anhui}
  \state{Hefei}
  \country{China}
}

\author{Xuanze Yang}
\affiliation{%
  \institution{University of Science and Technology of China}
  \city{Anhui}
  \state{Hefei}
  \country{China}
}

\author{Ruiqi Shu}
\affiliation{%
  \institution{Tsinghua University}
  \city{Beijing}
  \state{Haidian District}
  \country{China}
}

\author{Yuan Gao}
\affiliation{%
  \institution{Tsinghua University}
  \city{Beijing}
  \state{Haidian District}
  \country{China}
}

\author{Xiaomeng Huang}
\affiliation{%
  \institution{Tsinghua University}
  \city{Beijing}
  \state{Haidian District}
  \country{China}
}





\begin{abstract}
\input{component/1_abstract}
\end{abstract}

\vspace{-10pt}


\maketitle

\input{component/2_introduction}

\input{component/3_related_work}

\input{component/4_preliminary}

\input{component/5_methods}

\input{component/6_experiment}

\input{component/7_conclusion}

\clearpage
\normalem
\bibliographystyle{ACM-Reference-Format}
\bibliography{sample-base}

\appendix
\input{component/8_appendix}

\end{document}

%% file: component/1_abstract.tex
Neural operators have demonstrated considerable effectiveness in accelerating the solution of time-dependent partial differential equations (PDEs) by directly learning governing physical laws from data.
However, for PDEs governed by conservation laws(e.g., conservation of mass, energy, or matter), existing neural operators fail to satisfy conservation properties, which leads to degraded model performance and limited generalizability.
Moreover, we observe that distinct PDE problems generally require different optimal neural network architectures. This finding underscores the inherent limitations of specialized models in generalizing across diverse problem domains.
To address these limitations, we propose \textbf{E}xterior-Embedded \textbf{C}onservation \textbf{F}ramework (ECF), a universal conserving framework that can be integrated with various data-driven neural operators to enforce conservation laws strictly in predictions. 
The framework consists of two key components: a conservation quantity encoder that extracts conserved quantities from input data, and a conservation quantity decoder that adjusts the neural operator's predictions using these quantities to ensure strict conservation compliance in the final output. 
Since our architecture enforces conservation laws, we theoretically prove that it enhances model performance.
To validate the performance of our method, we conduct experiments on multiple conservation-law-constrained PDE scenarios, including adiabatic systems, shallow water equations, and the Allen-Cahn problem.
These baselines demonstrate that our method effectively improves model accuracy while strictly enforcing conservation laws in the predictions.

%% file: component/2_introduction.tex
\section{Introduction}
Partial differential equations (PDEs) provide the foundational framework for modeling multiscale phenomena across physics, chemistry, and biology. While traditional discretization methods (FDM, FVM, FEM)~\cite{godunov1959finite,dhatt2012finite,eymard2000finite,trefethen1996finite} remain indispensable, they suffer from inherent limitations: (1) the accuracy-speed tradeoff between fine and coarse grids, and (2) the need to re-run simulations for every small change—particularly problematic for large-scale applications like ocean forecasting. These limitations motivate the adoption of data-driven approaches for faster solutions.

Recent deep learning research has introduced a new paradigm for PDE modeling and prediction. For example, Neural operators can learn underlying physical relationships from data, enabling state prediction with reduced computational costs while achieving notable results in related studies. Unlike conventional deep learning applications, neural operator-based temporal prediction essentially learns nonlinear operator mappings between infinite-dimensional Banach spaces. 
This requires models to capture intrinsic dynamical system behaviors rather than performing simple data fitting.

However, current neural operators suffer from a fundamental limitation—physical non-conservation. For example, the Shallow Water equations~\cite{takamoto2022pdebench} represent a system where mass conservation must be preserved—the total water volume remains constant throughout the temporal evolution of water depth. Traditional numerical methods inherently incorporate prior knowledge of conservation laws through specialized algorithmic designs. Taking the Finite Volume Method (FVM) as an example, this approach partitions the computational domain into non-overlapping control volumes (grid cells), with its core principle rooted in conservation laws—the net flux of a physical quantity entering a control volume equals the rate of change of that quantity within the volume. This structural design fundamentally ensures strict adherence to conservation relationships throughout the solution process.

In contrast, current neural operator architectures lack dedicated structural designs for enforcing conservation laws. As illustrated in Figure~\ref{fig: visual result}, taking the CNO model as an example, we observe that its conserved quantity exhibits continuous temporal decay, indicating progressively accumulating conservation errors over time. This ultimately leads to significantly bluer regions in CNO's predictions compared to the ground truth (particularly in the central local area), causing the predicted solutions to deviate from physically realistic ones.

\begin{figure}
    \centering
    \subfigure[\normalsize Visual Prediction Comparison]{
    \includegraphics[width = 1 \linewidth]{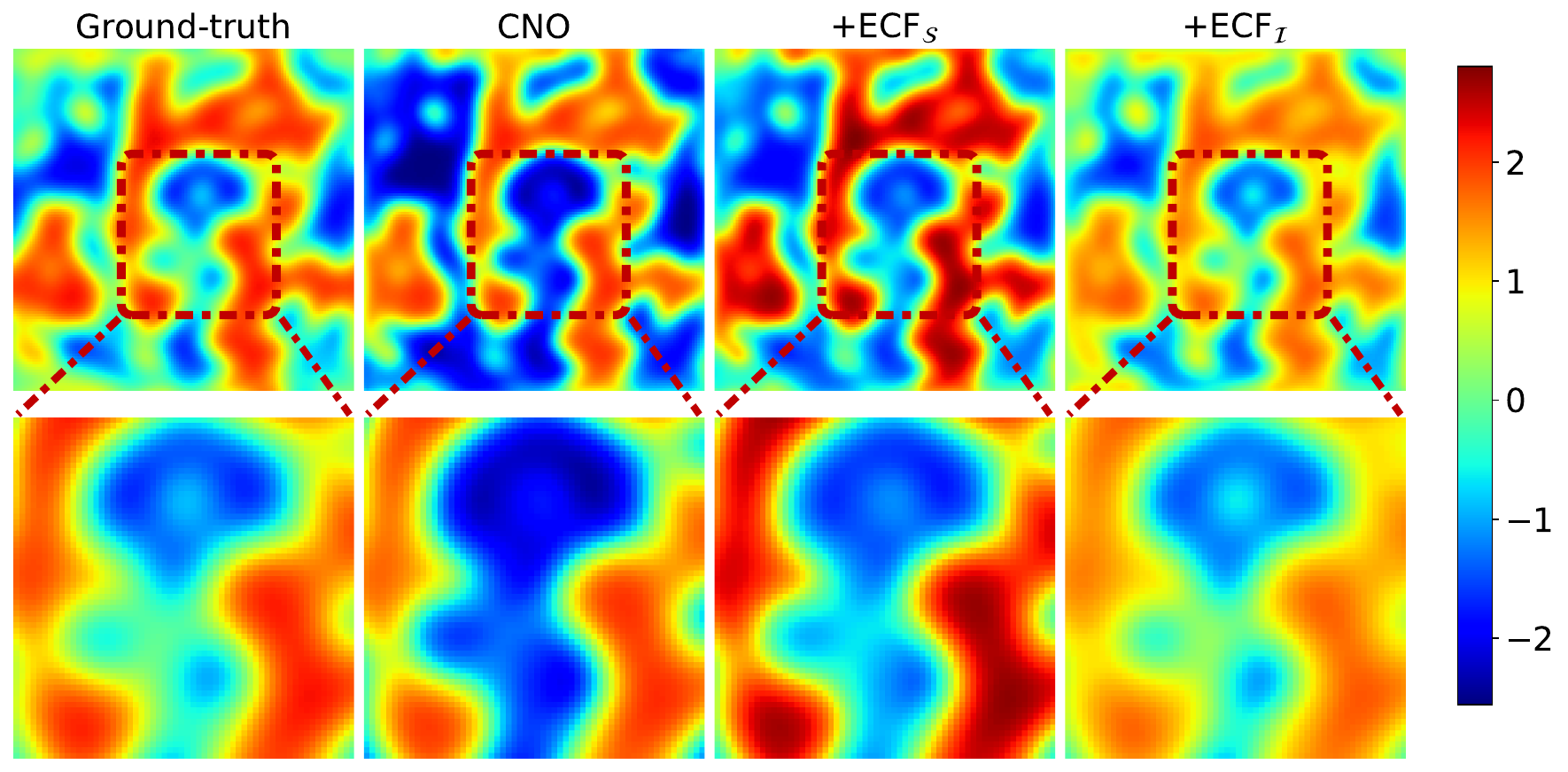}
    }\label{fig: a}
    \subfigure[\normalsize Conserved Quantity Comparison]{\includegraphics[width = 0.48\linewidth]{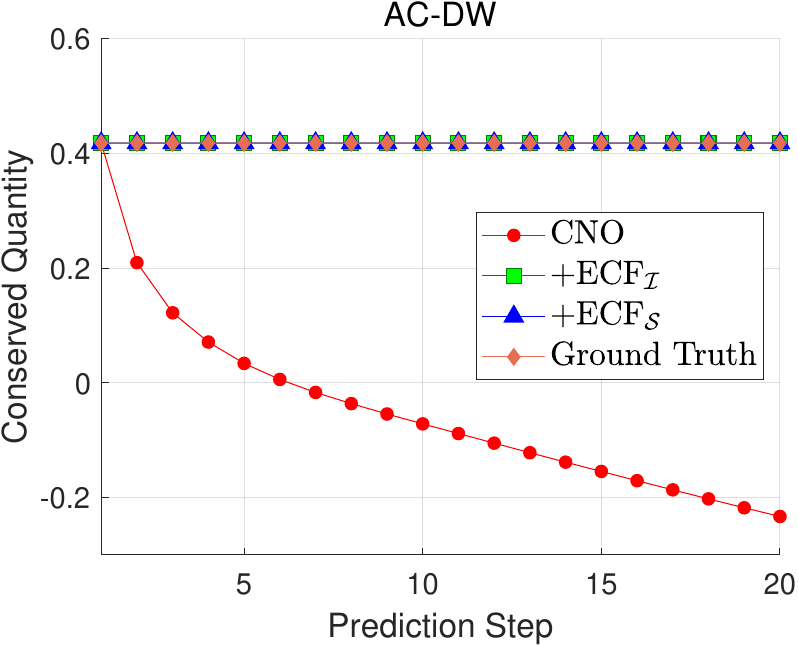}
    \includegraphics[width = 0.48\linewidth]{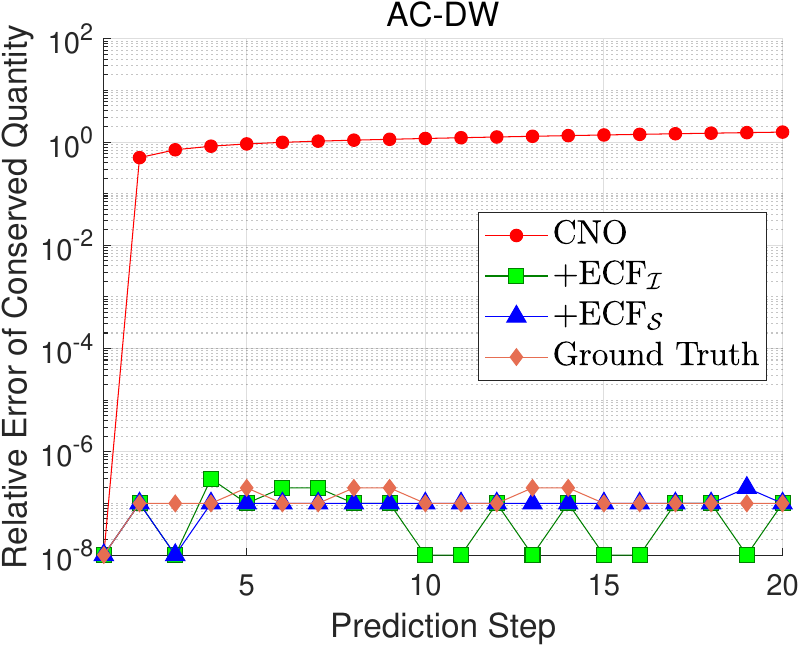}
    }\label{fig: b}
    \caption{(a) demonstrates the visualization of the CNO model's final prediction results on the Allen-Cahn equation (AC-DW) dataset. The first row presents the global prediction results, while the second row displays the corresponding central local region predictions, with deeper blue shades indicating smaller numerical values. (b) presents the predicted conserved quantities of both the CNO model and its corresponding embedded framework over 20 time steps, along with the relative error variations of these quantities compared to their initial states.}
    \label{fig: visual result}
\end{figure}
Notably, in the field of Neural Operators~\cite{hao2022physics, mcgreivy2024weak, zhang2025artificial}, we observe that different datasets often require distinct optimal neural operator architectures. As shown in Figure~\ref{fig: different neural operator result}, UNO achieves the best prediction performance for the Allen-Cahn problem, while FNO yields optimal results for adiabatic systems. 
This architectural dependence creates a fundamental challenge for conservation-preserving designs: specialized neural operators that enforce conservation laws may achieve excellent performance on their target systems, but their rigid architectures cannot adapt to different physical scenarios where alternative operator architectures perform better.
Therefore, we propose to develop a universal plug-and-play framework that can be directly embedded into various neural operators to enhance their performance.

We propose a universal framework called the \textbf{E}xterior-Embedded \textbf{C}onservation \textbf{F}ramework (ECF) to address these issues, which is compatible with various existing neural operators.
Our approach introduces two core innovations that fundamentally address the conservation-law challenge:
1) Conserved quantity Encoder: Extracts conservation-related information from input data through the Fourier transform.
2) Conserved quantity Decoder: Corrects the frequency-domain predictions of the neural operator in Fourier space using the conservation information obtained from the encoder, then applies the inverse Fourier transform to produce the final refined predictions.
To address potential gradient trajectory alterations introduced by the framework, we designed two training paradigms: Integrated Training Mode (ECF$_{\mathcal{I}}$) incorporates the correction framework as an end-to-end optimized component, and Staged Training Mode (ECF$_{\mathcal{S}}$) first trains the neural operator independently before applying the correction module.

Our framework constitutes a stable algorithm for reducing errors at specific frequencies without affecting predictions in other frequency bands, thereby theoretically guaranteeing strict error reduction in the predictions. Extensive experiments on benchmark datasets also validate the method's effectiveness. As demonstrated in Fig.~\ref {fig: visual result}, our proposed algorithm rigorously enforces conservation laws (with conservation errors limited only by machine precision), since our framework achieves superior prediction accuracy in the dark regions compared to the CNO model.

In summary, our contributions are as follows:

\begin{itemize}[leftmargin=*]
    \item \textbf{\textit{A Universal and Flexible Conservation Framework}:} We propose the \textbf{Exterior-Embedded Conservation Framework (ECF)}, a universal plug-and-play module that strictly enforces physical conservation laws in any neural operator via frequency-domain correction. Its design includes two flexible training paradigms, \textbf{Integrated Mode (ECF$_{\mathcal{I}}$)} and \textbf{Staged Mode (ECF$_{\mathcal{S}}$)}, making it adaptable to diverse optimization needs without altering the base operator's architecture.

    \item \textbf{\textit{Rigorous Theoretical Guarantees}}: Through frequency-domain analysis, we demonstrate the correlation between conservation quantity errors and relative Root Mean Square Error(RMSE), strictly reducing errors while preserving model expressiveness, providing a theoretical foundation for physical constraints in neural networks.

    \item \textbf{\textit{Comprehensive Experimental Validation}:} We conduct extensive experiments on a new \textit{benchmark} of $\textbf{6}$ distinct conservation law constrained PDE problems, which we developed to ensure standardized and fair evaluation. The results validate the superiority and generality of our framework, demonstrating substantial performance gains across all baselines, with improvements of up to \textbf{$37.7\%$}.
\end{itemize}

%% file: component/3_related_work.tex
\section{Related Work}
\textbf{Neural Operator.} Neural Operators represent an emerging mathematical modeling framework designed to learn the implicit Green's functions of partial differential equations (PDEs), providing a universal solution for characterizing spatiotemporal dynamical systems ~\cite{li2020neural,li2020fourier,wu2024transolver,raonic2023convolutional,rahman2022u,hao2023gnot}. Unlike traditional neural networks operating in finite-dimensional Euclidean spaces, it specializes in learning mappings between infinite-dimensional function spaces~\cite{li2020fourier}. This capability allows it to effectively capture high-order nonlinear relationships in input-output function pairs.
Several implementations with distinct advantages have emerged: Fourier Neural Operators (FNO) and variants achieve cross-resolution generalization through spectral integral operators~\cite{li2020fourier}; Convolutional Neural Operators (CNO) extract spatial features via local convolution kernels~\cite{raonic2023convolutional}; U-NO supports deep architectures with memory optimization~\cite{rahman2022u}; While Transolver innovatively employs physics-driven attention mechanisms to balance computational efficiency and physical consistency~\cite{wu2024transolver,hao2023gnot}. 
These models work well in some cases, but they don't have special parts built in to follow conservation rules exactly, so they can't guarantee conservation laws will always be obeyed.

\textbf{Learning hidden physics.} Learning and predicting complex physical phenomena directly from data is widespread across scientific and engineering applications~\cite{ghaboussi1998autoprogressive,ghaboussi1991knowledge,besnard2006finite,cai2021physics,carleo2019machine,he2021manifold,karniadakis2021physics,pfau2020ab,zhang2018deep}. The traditional finite volume method (FVM)~\cite{eymard2000finite} divides the computational domain into non-overlapping control volumes and assigns a node to each volume to store physical quantities. By integrating the governing conservation equations over each individual control volume, FVM transforms the principle of global conservation into local algebraic equations. This approach rigorously enforces conservation laws by ensuring that the net flux (inflows minus outflows) across each control volume's boundaries is balanced. In contrast, Physics-Informed Neural Networks (PINNs)~\cite{raissi2019physics} incorporate the governing physical equations as soft constraints within the neural network's loss function, integrating known physical laws into the learning process. Emerging PINN variants, such as the Allen-Cahn Neural Network (ACNN)~\cite{geng2024deep} for solving the Allen-Cahn equation and ClawNOs~\cite{liu2023harnessing} for divergence-free prediction, embed specific conservation laws as physical priors within this framework.
This study focuses on learning physical systems with conservation law constraints. Typical scenarios include but are not limited to: energy conservation in adiabatic systems and mass conservation in shallow water wave equations~\cite{caratheodory1909untersuchungen,takamoto2022pdebench}. Modeling such systems presents dual challenges: simultaneously achieving both data-fitting accuracy and strict compliance with physical constraints.

%% file: component/4_preliminary.tex
\section{Preliminary}
\subsection{Conservation Laws in Physical Systems}
The standard form of the conservation law of partial differential equations can be expressed as:
\begin{equation}
    \frac{\partial \bm{u}}{\partial t} + \nabla \cdot \bm{F}(\bm{u}) = \bm{S}(\bm{u},\bm{x},t), \quad \bm{x} \in \Omega, t \in [0,T],
    \label{eq:conservation_law_standard}
\end{equation}
where $\bm{u}(\bm{x},t)\in\mathbb{R}^d$ is the $d$-dimensional state vector representing conserved quantities, $\bm{F}(\bm{u})\in\mathbb{R}^{d\times m}$ is the flux tensor field characterizing transport processes, $\bm{S}(\bm{u},\bm{x},t)$ denotes source/sink terms accounting for external influences or internal reactions, and $\Omega\subset\mathbb{R}^m$ represents the $m$-dimensional spatial domain with boundary $\partial\Omega$.

Conservation laws fundamentally mean that the total amount of certain system quantities (like energy or mass) remains constant over time. These conserved quantities are defined as:
\begin{equation}
    \bm{E}(t) = \int_{\Omega}\bm{u}(\bm{x},t)d\bm{x}.
    \label{eq:conserved_quantity}
\end{equation}

Applying Gauss's divergence theorem, we obtain the following fundamental theorem(see Appendix~\ref{app: theorem proof} for complete proof):
\begin{theorem}\label{eq: theorem green}
The time derivative of the conserved quantity satisfies:
\begin{equation}
    \begin{aligned}
        \frac{d\bm{E}}{dt} = -\oint_{\partial\Omega}\bm{F}(\bm{u})\cdot\bm{n}dS + \int_{\Omega}\bm{S}d\bm{x},
    \end{aligned}
\end{equation}
where $\partial \Omega$ denotes the boundary of the domain $\Omega$ and $\bm{n}$ is the outward unit normal vector on $\partial\Omega$.
\end{theorem}
In Theorem~\ref{eq: theorem green}, the surface integral $\oint_{\partial\Omega} \bm{F}(\bm{u}) \cdot \bm{n} \, dS$ quantifies the net flux across the domain boundaries, encoding both convective and diffusive transport mechanisms. The volume integral $\int_{\Omega} \bm{S} \, d\bm{x}$ captures the net production (when $\bm{S}>0$) or dissipation (when $\bm{S}<0$) within the domain. This theorem reflects the fundamental physical principle that the temporal change of $\bm{E}$ equals the net inflow through boundaries plus the net internal generation.

Strict conservation requires satisfaction of two independent conditions:
\begin{itemize}
    \item The boundary condition $\bm{F}(\bm{u})\cdot\bm{n}|_{\partial\Omega}\equiv\bm{0}$ ensures no net flux exchange
    \item The source term condition $\int_{\Omega}\bm{S}d\bm{x}=\bm{0}$ guarantees global balance between internal generation and dissipation
\end{itemize}

We take adiabatic problems as an example, which satisfy the following equations:
\begin{equation}
    \begin{aligned}\label{eq:heat}
        &\frac{\partial \bm{u}}{\partial t} + \Delta \bm{u} = \bm{0}, \quad \bm{x} \in \Omega, t \in [0,T]\\
        &\nabla \bm{u} \cdot \bm{n} |_{\partial\Omega} = \bm{0}.
    \end{aligned}
\end{equation}

This system exhibits strict conservation because:
\begin{itemize}
    \item The Neumann boundary condition $\nabla\bm{u}\cdot\bm{n}|_{\partial\Omega}=\bm{0}$ satisfies the flux condition $\bm{F}(\bm{u})\cdot\bm{n}|_{\partial\Omega}\equiv\bm{0}$ (with $\bm{F}(\bm{u})=\nabla\bm{u}$)
    \item The homogeneous PDE ($\bm{S}\equiv\bm{0}$) trivially satisfies $\int_{\Omega}\bm{S}d\bm{x}=\bm{0}$
\end{itemize}

The conservation property is verified through:
\begin{equation}
    \frac{d\bm{E}}{dt} = -\underbrace{\oint_{\partial\Omega}\nabla\bm{u}\cdot\bm{n}dS}_{=0} + \underbrace{\int_{\Omega}\bm{0}\,d\bm{x}}_{=0} = 0,
\end{equation}
which implies:
\begin{equation}
    \bm{E}(t) \equiv \int_{\Omega}\bm{u}(\bm{x},t)d\Omega = \bm{E}(0) = \text{const}.
\end{equation}

\subsection{Neural Operators}
We consider a general form of temporal partial differential equations (PDEs). Let the variable $\bm{u}(\bm{x},t)\in \mathbb{R}^d$ defined on $\Omega\times T\subset  \mathbb{R}^{m+1}$ satisfy the following system~\cite{evans2022partial}:
\begin{equation}
    \begin{aligned}
    \frac{\partial \bm{u}}{\partial t} + \nabla \cdot \bm{F}(\bm{u}) &= \bm{S}(\bm{u},\bm{x},t), \quad &\bm{x} \in \Omega, t \in [0,T], \\
    \bm{u}(\bm{x},0) &= \bm{u}^0(\bm{x})  &\bm{x} \in \Omega,
    \end{aligned}
\end{equation}
where $\bm{F}(\bm{u})\in\mathbb{R}^{d\times m}$ is the flux tensor field characterizing transport processes, $\bm{S}(\bm{u},\bm{x},t)$ denotes source/sink terms accounting for external influences or internal reactions, $\bm{u}^0$ represents the initial condition. This class of initial value problems is general enough to encompass many fundamental PDEs, such as diffusion equations, shallow water equations, Allen-Cahn problems, etc. The neural operator $\mathcal{N}_\theta$ for predicting PDE solutions can be expressed as~\cite{kovachki2023neural}:
\begin{equation}
    \begin{aligned}
        \mathcal{N}_\theta:\ &\mathcal{U}\rightarrow \mathcal{V}
        \\&\bm{u}^{t}\rightarrow \bm{u}^{t+1},
    \end{aligned}
\end{equation}
where $\mathcal{U}$ and $\mathcal{V}$ are the Sobolev spaces containing the solution functions $\bm{u}^{t}$ and $\bm{u}^{t+1}$, respectively.
In practical computations, the solution function $\bm{u}$ is discretized on spatiotemporal grids as $\bm{u}=(\bm{u}^{0},\ldots,\bm{u}^{T})$, where $\bm{u}^{t}=\{\bm{u}(\bm{x},t):\bm{x}\in \Omega \}$ represents the spatial field distribution at time step $t$ (with $t$ denoting the temporal iteration index rather than specific time values), for $1\leq t\leq T$.
Within the optimization framework based on a finite-dimensional parameter space $\Theta$, we can determine the optimal parameter $\theta^{*} \in \Theta$ by defining a loss function and implementing specific optimization algorithms that leverage observed grid data $\bm{u}^{t}$ and its evolved state $\bm{u}^{t+1}$. The optimal parameter $\theta^*$ satisfies $\mathcal{N}_{\theta^*}(\bm{u}^{t})\approx \bm{u}^{t+1}$, thereby constructing a learning paradigm in an infinite-dimensional space.

These neural operators aim to learn infinite-dimensional operators that describe the evolution of physical variables, thereby achieving generalized modeling of the entire physical system in Banach space.

%% file: component/5_methods.tex
\section{Method}
\subsection{Motivation}
\begin{figure}
    \centering
    \includegraphics[width=1\linewidth]{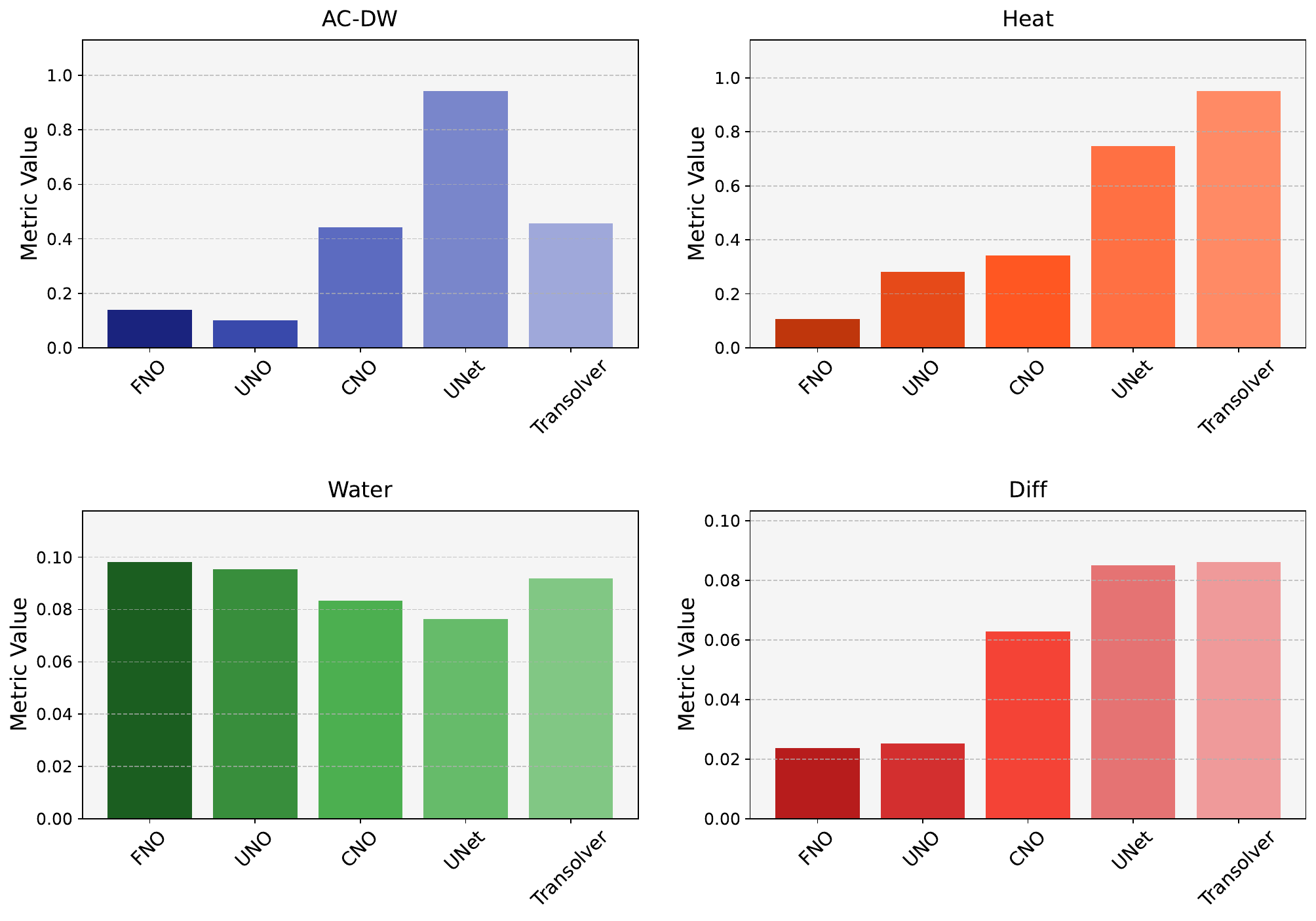}
    \caption{The predictive results of five models on four different datasets were evaluated using relative mean squared error (RMSE), where smaller errors indicate better model performance.}
    \label{fig: different neural operator result}
\end{figure}

As demonstrated in Figure~\ref{fig: visual result}, current neural operators fail to satisfy conservation laws when solving conservation-critical scenarios. To address this limitation, we propose to develop a novel framework that enables models to rigorously capture conservation law information. Furthermore, experimental results in Figure~\ref{fig: different neural operator result} reveal that the optimal neural operator varies significantly across different problems. Thus, our goal is to design a universal framework capable of augmenting any neural operator while strictly enforcing conservation law constraints.

For clarity, we consider a PDE problem defined on the $m$-dimensional domain $\Omega = [0,L]^m$. The ground truth $\hat{\bm{v}}(\bm{x},t)$ can be expressed as a multidimensional Fourier series:
\begin{equation}
\begin{aligned}\label{eq:Fourier}
    \hat{\bm{v}}(\bm{x},t) = \sum_{\bm{n}\in\mathbb{Z}^m} \hat{\bm{c}}_{\bm{n}}(t) e^{i\frac{2\pi}{L}\bm{n}\cdot\bm{x}},
\end{aligned}
\end{equation}
where $\bm{n} = (n_1,...,n_m)$ is the multi-index of frequency modes, and the coefficients $\hat{\bm{c}}_{\bm{n}}(t)$ are given by:
\begin{equation}
\begin{aligned}
\hat{\bm{c}}_{\bm{n}}(t) = \frac{1}{L^m}\int_\Omega \hat{\bm{v}}(\bm{x},t) e^{-i\frac{2\pi}{L}\bm{n}\cdot\bm{x}} d\bm{x}.
\end{aligned}
\end{equation}
Here, $\hat{\bm{c}}_{\bm{n}}(t)$ represents the complex Fourier coefficient corresponding to frequency mode $\bm{n}$, capturing the amplitude and phase information of each spectral component. The zero-frequency term $\hat{\bm{c}}_\mathbf{0}(t)$, obtained when $\bm{n}=\mathbf{0}$, corresponds to the spatial average of the solution over the domain $\Omega$ and plays a crucial role in conservation laws. 
Specifically, the total conserved quantity $\hat{\bm{E}}(t)$ is closely connected with the zero-frequency component $\hat{\bm{c}}_{\mathbf{0}}(t)$:
\begin{equation}
\begin{aligned}
\hat{\bm{E}}(t) = \int_\Omega \hat{\bm{v}}(\bm{x},t) d\bm{x} = \hat{\bm{c}}_{\mathbf{0}}(t) L^m,
\end{aligned}
\end{equation}
as all oscillatory terms ($\bm{n}\neq\mathbf{0}$) integrate to zero over $\Omega$. The conservation law requires:
\begin{equation}\label{eq:a0 equality}
\hat{\bm{E}}(t) = \hat{\bm{c}}_{\mathbf{0}}(t) L^m = \text{constant} \implies \hat{\bm{c}}_{\mathbf{0}}(t) = \hat{\bm{c}}_{\mathbf{0}}(t+1)=\text{constant}.
\end{equation}
For the solution $\bm{v}(\bm{x},t+1)=\mathcal{N}_\theta(\hat{\bm{v}}^t)$ predicted by the neural operator, we perform an analogous Fourier decomposition:
\begin{equation}
\begin{aligned}
\bm{v}(\bm{x},t+1) = \sum_{\bm{n}\in\mathbb{Z}^m} \bm {c}_{\bm{n}}(t+1) e^{i\frac{2\pi}{L}\bm{n}\cdot\bm{x}},
\end{aligned}
\end{equation}
where the coefficients $\bm{c}_{\bm{n}}(t+1)$ are computed through the $m$-dimensional Fourier transform of the neural operator's output.

To enforce conservation, we correct the predicted solution by preserving the zero-frequency mode from the input:
\begin{equation}
\begin{aligned}\label{eq:Corrected conserved quantities}
\bar{\bm{v}}(\bm{x},t+1) = \hat{\bm{c}}_{\mathbf{0}}(t) + \sum_{\bm{n}\in\mathbb{Z}^m\setminus\{\mathbf{0}\}} \bm{c}_{\bm{n}}(t+1) e^{i\frac{2\pi}{L}\bm{n}\cdot\bm{x}}.
\end{aligned}
\end{equation}

The conserved quantity of the corrected solution then satisfies:
\begin{equation}\label{eq:integration unchanged}
\begin{aligned}
\bar{\bm{E}}(t+1) &= \int_\Omega \bar{\bm{v}}(\bm{x},t+1) d\bm{x} \\
&= \hat{\bm{c}}_{\mathbf{0}}(t)L^m + \sum_{\bm{n}\neq\mathbf{0}} \bm{c}_{\bm{n}}(t+1) \underbrace{\int_\Omega e^{i\frac{2\pi}{L}\bm{n}\cdot\bm{x}} d\bm{x}}_{=0} \\
&= \hat{\bm{E}}(0).
\end{aligned}
\end{equation}

This indicates that the corrected solution $\bar{\bm{v}}^{t+1}$ at the current time step strictly satisfies the conservation property. Essentially, the neural operator predicts the Fourier coefficients for each frequency mode. Our correction enforces energy conservation by setting the zero-frequency term to its ground truth coefficient $\hat{\bm{c}}_{\mathbf{0}}(t)$, while preserving all higher-frequency components ($\bm{c}_{\bm{n}}(t+1)$ for $\bm{n}\neq\mathbf{0}$). We analyze the relationship between Fourier expansion coefficients and prediction errors, as well as the error correction effects of Equation~\ref{eq:Corrected conserved quantities}, through the following theorems:

\begin{theorem}\label{thorem: frequent and error}
    The error between the predicted result $\bm{v}^{t+1}$ and the ground truth $\hat{\bm{v}}^{t+1}$ can be expressed in terms of their Fourier coefficients. Specifically:
    \begin{equation}
        \begin{aligned}
            &\|\bm{v}^{t+1} - \hat{\bm{v}}^{t+1}\|_{L^2(\Omega)}^2\\ 
            &= L^m \sum_{\bm{n}\in\mathbb{Z}^m} |\bm{c}_{\bm{n}}(t+1) - \hat{\bm{c}}_{\bm{n}}(t+1)|^2 \\
            &= L^m \Big[ |\bm{c}_{\mathbf{0}}(t+1) - \hat{\bm{c}}_{\mathbf{0}}(t+1)|^2 \\
            &\ \quad + \sum_{\bm{n}\in\mathbb{Z}^m\setminus\{\mathbf{0}\}} |\bm{c}_{\bm{n}}(t+1) - \hat{\bm{c}}_{\bm{n}}(t+1)|^2 \Big],
        \end{aligned}
    \end{equation}
    where $\|\cdot\|_{L^2(\Omega)}$ represents the $L^2$-norm on domain $\Omega$, defined as $\| f\|_{L^2(\Omega)} = \left(\int_{\Omega}|f(\bm{x})|^2d\bm{x}\right)^{\frac{1}{2}}$. 
\end{theorem}

\begin{figure*}
    \centering
    \includegraphics[width=1\linewidth]{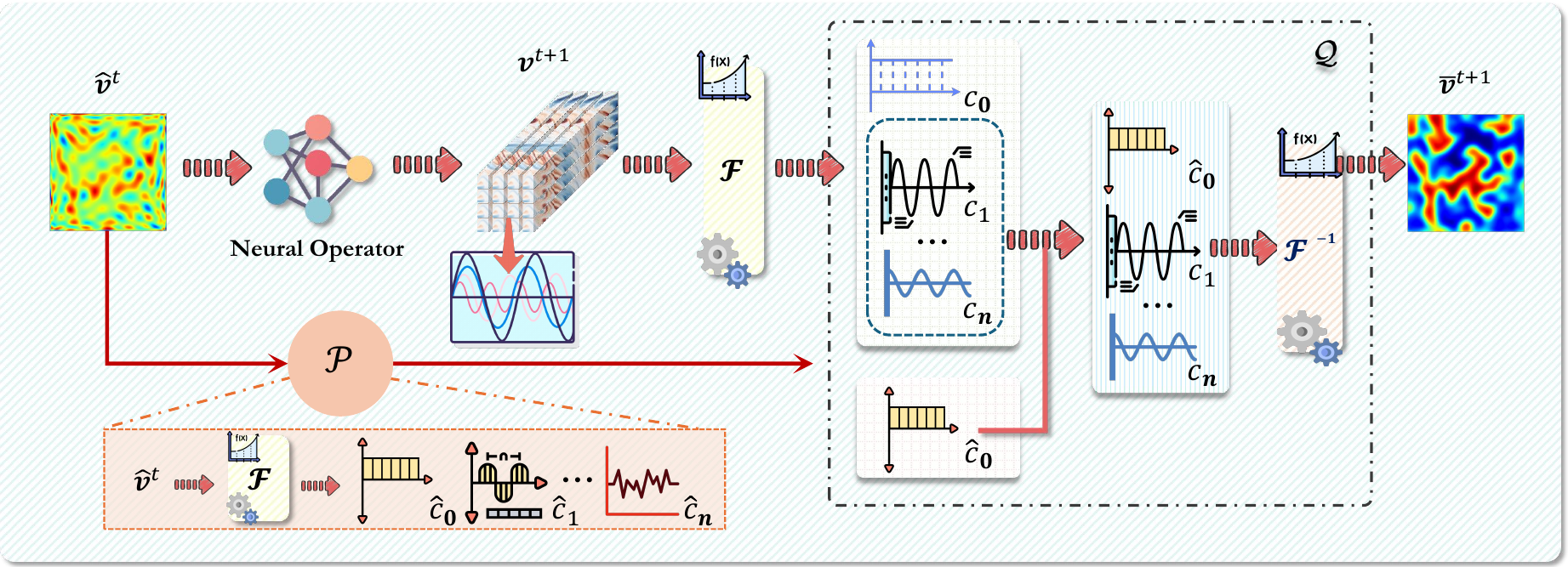}
    \caption{\underline{Overview of the \textit{ECF} architecture}. The \textit{Conserved Quantity Encoder} $\mathcal{P}$ extracts the zero-frequency signal $\hat{\bm{c}}_{\bm{0}}$ from input data via Fourier transform $\mathcal{F}$. The \textit{Conserved Quantity Decoder} $\mathcal{Q}$ replaces $\bm{c}_{\bm{0}}$ in the Neural Operator's predictions with $\hat{\bm{c}}_{\bm{0}}$ and obtains the final output through inverse Fourier transform $\mathcal{F}^{-1}$. The \textit{Neural Operator} in the diagram can represent various types of neural operators.
}
    \label{fig: method}
\end{figure*}

Theorem~\ref{thorem: frequent and error} demonstrates that the total error can be decomposed into a weighted sum of squared errors across all frequency modes, with each Fourier coefficient contributing independently. This implies that modifying the zero-frequency error (i.e., $\bm{c}_{\mathbf{0}}$) does not affect the prediction results at other frequencies.
\begin{theorem}\label{theorem 2}
Let $\hat{\bm{v}}^{t+1}$, $\bm{v}^{t+1}$, and $\bar{\bm{v}}^{t+1}$ denote the ground truth at time $t+1$, the solution function predicted by the neural operator, and the corrected solution function obtained via Equation \ref{eq:Corrected conserved quantities}, respectively. Then:
\begin{equation}
\| \bar{\bm{v}}^{t+1} - \hat{\bm{v}}^{t+1} \|_{L^2(\Omega)} \leq \| \bm{v}^{t+1} - \hat{\bm{v}}^{t+1} \|_{L^2(\Omega)}.
\end{equation}
The equality holds if and only if:
\begin{equation}
    \int_\Omega \bm{v}^{t+1}d\bm{x} = \int_\Omega \hat{\bm{v}}^{t+1} d\bm{x}.
\end{equation}
\end{theorem}

Theorem~\ref{theorem 2} demonstrates that our method can rigorously reduce the error of predictions that violate conservation laws. We provide detailed proofs of the above two theorems in Appendix~\ref{app: theorem proof}.

\subsection{Exterior-Embedded Conservation Framework}

Based on the correction scheme presented in Equation~\ref{eq:Corrected conserved quantities}, we propose a framework that can be exterior-embedded into various neural operators to enforce conservation laws, called \textbf{E}xterior-Embedded \textbf{C}onservation \textbf{F}ramework (ECF). As shown in Figure~\ref{fig: method}, our framework consists of three core components: (1) a conserved quantity encoder $\mathcal{P}:\ \mathcal{U} \rightarrow \mathbb{R}^d$ that extracts conserved quantities from initial states, (2) a time evolution operator $\mathcal{N}:\ \mathcal{U}\rightarrow \mathcal{V}$ that learns the evolution laws of physical systems in latent space, (3) a conserved quantity decoder $\mathcal{Q}:\ \mathcal{V}\times \mathbb{R}^d\rightarrow \mathcal{V}$ that corrects neural operator outputs using conserved quantities. Specifically, the model structure can be expressed as:
\begin{equation}
    \begin{aligned}
        \bar{\bm{v}}^{t+1} = \mathcal{Q}(\mathcal{N}(\hat{\bm{v}}^{t}),\mathcal{P}(\hat{\bm{v}}^{t})).
    \end{aligned}
\end{equation}

\subsubsection{Conserved Quantity Encoder}

The conserved quantity encoder is a module designed to capture conserved quantities from data. Specifically:
\begin{equation}
    \begin{aligned}
        \mathcal{P} = \mathcal{T} \circ \mathcal{F},
    \end{aligned}
\end{equation}
where $\mathcal{F}$ is a truncated Fourier transform operator that maps the solution function $\bm{v}^{t}$ in the spatial domain to the frequency domain. For $m$-dimensional problems:
\begin{equation}
    \begin{aligned}
        \mathcal{F}:\ &\mathcal{U}\rightarrow \mathbb{C}^{N_1\times\cdots\times N_m}\\
        & \hat{\bm{v}}^{t} \mapsto \{\hat{\bm{c}}_{\bm{n}}(t)\}_{\bm{n}\in\mathcal{I}},
    \end{aligned}
\end{equation}
where $\mathcal{I} = \{\bm{n}\in\mathbb{Z}^m : \|\bm{n}\|_\infty \leq N\}$ is the frequency index set, and $\hat{\bm{c}}_{\bm{n}}(t)$ are the Fourier coefficients (see Equation~\ref{eq:Fourier}). The truncation operator $\mathcal{T}$ extracts the zero-frequency component:
\begin{equation}
    \begin{aligned}
        \mathcal{T}:\  &\mathbb{C}^{N_1\times\cdots\times N_m}\rightarrow \mathbb{C}\\
        & \{\hat{\bm{c}}_{\bm{n}}(t)\} \mapsto \hat{\bm{c}}_{\mathbf{0}}(t).
    \end{aligned}
\end{equation}
As shown in Equation~\ref{eq:Corrected conserved quantities}, the conservation law requires $\bm{c}_{\mathbf{0}}(t)$ to remain constant, making it the conserved quantity for correction.

\subsubsection{Conserved Quantity Decoder}
The decoder corrects the zero-frequency term in the Fourier transform of predictions:
\begin{equation}
    \begin{aligned}
        \mathcal{Q} = \mathcal{F}^{-1}\circ \mathcal{C}\circ (\mathcal{F}(\bm{v}^{t+1}),\hat{\bm{c}}_{\mathbf{0}}(t)),
    \end{aligned}
\end{equation}
where $\bm{v}^{t+1} = \mathcal{N}_\theta(\bm{v}^{t})$ is the neural operator's prediction,
$\hat{\bm{c}}_{\mathbf{0}}(t) = \mathcal{P}(\hat{v}^{t})$ is the zero-frequency component extracted by the conserved quantity encoder,
and $\mathcal{C}$ is the high-dimensional correction operator:
\begin{equation}
    \begin{aligned}
        \mathcal{C}:\ &\mathbb{C}^{N_1\times\cdots\times N_m}\times\mathbb{C}\rightarrow\mathbb{C}^{N_1\times\cdots\times N_m}\\
        & (\{\bm{c}_{\bm{n}}(t+1)\}, \hat{\bm{c}}_{\mathbf{0}}(t)) \mapsto \{\bm{c}'_{\bm{n}}\}\\
        & \text{with } \bm{c}'_{\bm{n}} = \begin{cases}
            \hat{\bm{c}}_{\mathbf{0}}(t) & \text{if } \bm{n}=\mathbf{0}\\
            \bm{c}_{\bm{n}}(t+1) & \text{otherwise}.
        \end{cases}
    \end{aligned}
\end{equation}
This operation preserves all non-zero frequency modes while enforcing conservation through the zero-frequency term.

\subsubsection{Training Type}\label{sec: training type}
Although Theorem~\ref{theorem 2} theoretically guarantees error reduction through correction, the framework's introduction may alter loss landscape dynamics during training. To address potential gradient trajectory modifications while accommodating different error regimes, we propose two specialized training paradigms:

\begin{enumerate}
    \item \textbf{Integrated Training Mode (+ECF$_{\mathcal{I}}$)}: Designed for scenarios with significant conservation quantity errors, this end-to-end approach incorporates the correction framework directly into the optimization process:
    \begin{equation}
        Loss_{\mathcal{I}} = \|\mathcal{Q}(\mathcal{N}(\hat{\bm{v}}^{t}),\mathcal{P}(\hat{\bm{v}}^{t}))-\hat{\bm{v}}^{t+1}\|_{L^2(\Omega)}^2.
    \end{equation}
    The framework actively guides the neural operator's training through gradient feedback from both $\mathcal{N}$ and $\mathcal{Q}$ components.

    \item \textbf{Staged Training Mode (+ECF$_{\mathcal{S}}$)}: Optimized for stable fine-tuning scenarios with minor conservation errors, this decoupled approach first trains the base operator:
    \begin{equation}
        Loss_{\mathcal{S}} = \|\mathcal{N}(\hat{\bm{v}}^{t})-\hat{\bm{v}}^{t+1}\|_{L^2(\Omega)}^2\ .
    \end{equation}
    followed by the subsequent integration of the correction module. This preserves the original operator's training dynamics while guaranteeing strict error reduction through Theorem~\ref{theorem 2}.
\end{enumerate}
The two paradigms fundamentally represent distinct strategies for spectral editing: +ECF$_{\mathcal{I}}$ achieves collaborative reconstruction of the frequency spectrum, while +ECF$_{\mathcal{S}}$ adopts isolated processing of spectral components. This difference manifests in the Fourier domain as follows: the former induces renormalization of spectral energy distribution, whereas the latter preserves the relative relationship between fundamental frequencies and harmonics.

%% file: component/6_experiment.tex
\section{Experiment}
In this section, we first introduce our datasets, experimental settings, and neural operators used. Then we demonstrate the capability of our framework to learn multiple partial differential equations through extensive experiments. Finally, we conduct additional experiments to further illustrate the performance characteristics of our framework.

\subsection{Setup}
\subsubsection{\textbf{Dataset}}\label{sec: dataset}
As shown in Table~\ref{tab: main experiment}, we use six common PDE problems as our benchmark datasets. Below is an introduction to these datasets, with detailed implementation specifics provided in Appendix~\ref{app: dataset}:
\begin{itemize}[leftmargin=*]
    \item \textbf{Allen-Cahn Equation}~\cite{geng2024deep} is a classical stiff semilinear parabolic equation describing phase separation and phase transition processes in phase-field modeling of multicomponent physical systems. The conservative Allen-Cahn equation is an improved version of the classical equation that additionally maintains mass conservation. This equation contains a nonlinear potential function. In this paper, we select two common potential functions for the dataset: double-well potential and Flory-Huggins potential, recorded in the table as datasets AC-DW and AC-FH, respectively.
    \item \textbf{Shallow-Water Equation}~\cite{takamoto2022pdebench} is derived from the compressible Navier-Stokes equations and provide a suitable framework for simulating free-surface flow problems. This dataset is labeled as "Water" in the table.
    \item  \textbf{Adiabatic Process}~\cite{caratheodory1909untersuchungen} describes a thermodynamic process in an adiabatic system, where no heat or particle exchange occurs with the external environment, making it a type of closed system. This dataset is labeled as "Heat" in the table.
    \item  \textbf{Diffusion Equation}~\cite{mcrae1982numerical} is a parabolic partial differential equation that describes how particles diffuse under a concentration gradient. According to Fick's law, the diffusion flux is proportional to the concentration gradient. This dataset is labeled as "Diff" in the table.
    \item \textbf{Convection-Diffusion Equation}~\cite{gupta1984single} is a parabolic partial differential equation combining the diffusion and convection (advection) equations. It describes the transport of particles, energy, or other physical quantities within a physical system due to two processes: diffusion and convection. This dataset is labeled as "CD" in the table.
\end{itemize}

\subsubsection{\textbf{Train and Test}}\label{sec: Train and Test}
For all models, we employ the AdamW optimizer with a learning rate of $1\times10^{-3}$ and train them for $1,000$ epochs, with the results averaged over five fixed random seeds. The models were trained on servers equipped with 3090 GPUs. Unless otherwise stated, we use the standard RMSE (Root Mean Square Error) to measure prediction quality. The details of model training and testing are provided in Appendix~\ref{app: training details}.

\subsubsection{\textbf{Baselines}}\label{sec: Baselines}
We evaluate our method using leading models from neural operators, computer vision, and time series prediction. The neural operator models FNO~\cite{li2020fourier}, UNO~\cite{rahman2022u}, and CNO~\cite{raonic2023convolutional} are top choices for Banach space mapping tasks. For computer vision, we use the widely adopted UNet~\cite{ronneberger2015u}, which performs well across various applications. In time series analysis, we employ Transolver~\cite{wu2024transolver}, a model specifically designed for PDE problems. The model sizes are specified in Appendix~\ref{app: model details}.

\subsection{Main Results}
\begin{table*}[h!]
\centering
\caption{In all scenarios, we integrate our framework into existing baselines using the two training paradigms mentioned in Section~\ref{sec: training type}: +ECF$_{\mathcal{I}}$ and +ECF$_{\mathcal{S}}$ for comparative evaluation. The Root Mean Square Error (RMSE) is adopted as the evaluation metric, where lower values indicate higher prediction accuracy.
Cells highlighted in blue represent the best-performing method among the three alternatives for each baseline on the corresponding dataset: (1) the original baseline, (2) baseline+ECF$_{\mathcal{I}}$, and (3) baseline+ECF$_{\mathcal{S}}$.}
\label{tab: main experiment}
\small
\renewcommand{\arraystretch}{1.1}
\begin{tabular}{l|cccccc}
\hline
Model & AC-DW & AC-FH & Heat & Water & Diff & CD \\
\hline
FNO & $1.40\text{E-}01 \pm 4.38\text{E-}03$ & $1.54\text{E-}01 \pm 1.12\text{E-}02$ & $1.07\text{E-}01 \pm 6.21\text{E-}03$ & $9.81\text{E-}02 \pm 5.63\text{E-}03$ & $2.37\text{E-}02 \pm 1.31\text{E-}03$ & $2.93\text{E-}01 \pm 2.26\text{E-}02$ \\
+ECF$_{\mathcal{I}}$ & \cellcolor{blue!15}$1.26\text{E-}01 \pm 4.17\text{E-}03$ & \cellcolor{blue!15}$1.44\text{E-}01 \pm 8.27\text{E-}03$ & \cellcolor{blue!15}$8.94\text{E-}02 \pm 9.15\text{E-}03$ & \cellcolor{blue!15}$8.81\text{E-}02 \pm 1.22\text{E-}03$ & \cellcolor{blue!15}$1.60\text{E-}02 \pm 8.66\text{E-}04$ & \cellcolor{blue!15}$2.50\text{E-}01 \pm 3.75\text{E-}02$ \\
+ECF$_{\mathcal{S}}$ & $1.31\text{E-}01 \pm 5.64\text{E-}03$ & $1.49\text{E-}01 \pm 1.10\text{E-}02$ & $1.03\text{E-}01 \pm 6.90\text{E-}03$ & $9.59\text{E-}02 \pm 4.37\text{E-}03$ & $2.18\text{E-}02 \pm 1.06\text{E-}03$ & $2.90\text{E-}01 \pm 2.18\text{E-}02$ \\
\hline
UNO & $1.00\text{E-}01 \pm 7.00\text{E-}03$ & $9.69\text{E-}02 \pm 6.01\text{E-}03$ & $2.81\text{E-}01 \pm 5.32\text{E-}03$ & $9.54\text{E-}02 \pm 9.77\text{E-}03$ & $2.52\text{E-}02 \pm 1.02\text{E-}03$ & $1.41\text{E-}01 \pm 7.41\text{E-}03$ \\
+ECF$_{\mathcal{I}}$ & \cellcolor{blue!15}$6.23\text{E-}02 \pm 6.40\text{E-}03$ & \cellcolor{blue!15}$8.09\text{E-}02 \pm 5.41\text{E-}03$ & \cellcolor{blue!15}$2.59\text{E-}01 \pm 9.12\text{E-}03$ & $9.68\text{E-}02 \pm 7.29\text{E-}03$ & $2.52\text{E-}02 \pm 6.88\text{E-}04$ & \cellcolor{blue!15}$1.28\text{E-}01 \pm 1.24\text{E-}02$ \\
+ECF$_{\mathcal{S}}$ & $8.60\text{E-}02 \pm 6.06\text{E-}03$ & $9.55\text{E-}02 \pm 5.30\text{E-}03$ & $2.77\text{E-}01 \pm 4.51\text{E-}03$ & \cellcolor{blue!15}$9.36\text{E-}02 \pm 9.78\text{E-}03$ & \cellcolor{blue!15}$2.31\text{E-}02 \pm 9.79\text{E-}04$ & $1.39\text{E-}01 \pm 7.04\text{E-}03$ \\
\hline
CNO & $4.43\text{E-}01 \pm 3.80\text{E-}02$ & $8.18\text{E-}01 \pm 2.34\text{E-}01$ & $3.43\text{E-}01 \pm 1.89\text{E-}02$ & $8.34\text{E-}02 \pm 1.60\text{E-}02$ & $6.29\text{E-}02 \pm 7.62\text{E-}03$ & $1.02\text{E+}00 \pm 1.56\text{E-}01$ \\
+ECF$_{\mathcal{I}}$ & \cellcolor{blue!15}$3.12\text{E-}01 \pm 3.35\text{E-}02$ & \cellcolor{blue!15}$5.04\text{E-}01 \pm 5.36\text{E-}02$ & $2.50\text{E-}01 \pm 5.09\text{E-}02$ & \cellcolor{blue!15}$7.16\text{E-}02 \pm 1.27\text{E-}02$ & \cellcolor{blue!15}$5.34\text{E-}02 \pm 2.95\text{E-}03$ & \cellcolor{blue!15}$8.24\text{E-}01 \pm 5.39\text{E-}02$ \\
+ECF$_{\mathcal{S}}$ & $3.38\text{E-}01 \pm 1.97\text{E-}02$ & $7.02\text{E-}01 \pm 1.88\text{E-}01$ & \cellcolor{blue!15}$2.36\text{E-}01 \pm 1.16\text{E-}02$ & $8.21\text{E-}02 \pm 1.63\text{E-}02$ & $6.12\text{E-}02 \pm 7.66\text{E-}03$ & $9.27\text{E-}01 \pm 1.47\text{E-}01$ \\
\hline
UNet & $9.41\text{E-}01 \pm 6.93\text{E-}02$ & $9.95\text{E-}01 \pm 8.46\text{E-}02$ & $7.47\text{E-}01 \pm 7.66\text{E-}02$ & $7.63\text{E-}02 \pm 1.11\text{E-}02$ & $8.50\text{E-}02 \pm 2.03\text{E-}02$ & $1.13\text{E+}00 \pm 1.61\text{E-}01$ \\
+ECF$_{\mathcal{I}}$ & $7.87\text{E-}01 \pm 2.61\text{E-}01$ & \cellcolor{blue!15}$7.21\text{E-}01 \pm 1.31\text{E-}01$ & \cellcolor{blue!15}$4.31\text{E-}01 \pm 5.88\text{E-}02$ & $9.16\text{E-}02 \pm 1.38\text{E-}02$ & $8.46\text{E-}02 \pm 1.56\text{E-}02$ & $1.06\text{E+}00 \pm 1.98\text{E-}01$ \\
+ECF$_{\mathcal{S}}$ & \cellcolor{blue!15}$6.92\text{E-}01 \pm 7.80\text{E-}02$ & $8.04\text{E-}01 \pm 1.16\text{E-}01$ & $5.23\text{E-}01 \pm 6.04\text{E-}02$ & \cellcolor{blue!15}$7.52\text{E-}02 \pm 1.13\text{E-}02$ & \cellcolor{blue!15}$8.31\text{E-}02 \pm 2.03\text{E-}02$ & \cellcolor{blue!15}$9.77\text{E-}01 \pm 1.65\text{E-}01$ \\
\hline
Transolver & $4.56\text{E-}01 \pm 4.02\text{E-}02$ & $4.30\text{E-}01 \pm 1.31\text{E-}02$ & $9.50\text{E-}01 \pm 1.14\text{E-}02$ & $9.19\text{E-}02 \pm 6.60\text{E-}04$ & $8.61\text{E-}02 \pm 6.90\text{E-}04$ & $7.24\text{E-}01 \pm 4.49\text{E-}02$ \\
+ECF$_{\mathcal{I}}$ & \cellcolor{blue!15}$4.11\text{E-}01 \pm 1.31\text{E-}02$ & $4.51\text{E-}01 \pm 1.58\text{E-}02$ & \cellcolor{blue!15}$8.95\text{E-}01 \pm 8.93\text{E-}03$ & \cellcolor{blue!15}$8.51\text{E-}02 \pm 1.56\text{E-}03$ & \cellcolor{blue!15}$8.50\text{E-}02 \pm 8.10\text{E-}04$ & $8.78\text{E-}01 \pm 1.58\text{E-}01$ \\
+ECF$_{\mathcal{S}}$ & $4.29\text{E-}01 \pm 4.25\text{E-}02$ & \cellcolor{blue!15}$4.13\text{E-}01 \pm 1.12\text{E-}02$ & $9.22\text{E-}01 \pm 2.12\text{E-}02$ & $8.91\text{E-}02 \pm 2.40\text{E-}03$ & $8.51\text{E-}02 \pm 1.05\text{E-}03$ & \cellcolor{blue!15}$7.19\text{E-}01 \pm 4.32\text{E-}02$ \\
\hline
\end{tabular}

\end{table*}
Table~\ref{tab: main experiment} presents the root mean square error (RMSE) results for our method on various baselines, both without and with our framework applied across different test datasets. The first row indicates the datasets used; each subsequent group of three rows shows: the original baseline performance, results from embedding our framework during its training (+ECF$_{\mathcal{S}}$), and performance when embedding it post-training into pre-existing models (+ECF$_{\mathcal{S}}$).

As shown in Table~\ref{tab: main experiment}, our framework consistently improves model prediction results. The +ECF$_{\mathcal{I}}$ embedding approach enhances performance across most baselines and datasets, as evidenced by its $37.7\%$ error reduction for UNO on AC-DW problems and $32.5\%$ improvement for FNO on Diff problems, demonstrating clear superiority and strong generalization capability. Although +ECF$_{\mathcal{I}}$ occasionally causes minor performance degradation (e.g., a $4.88\%$ error increase for Transolver on AC-FH problems), these cases involve minimal error margins. We observe that +ECF$_{\mathcal{S}}$ also provides limited improvement for Transolver on AC-FH, suggesting the baseline already achieves low conservation error in this scenario, leaving little room for further enhancement. Additionally, as discussed in Section~\ref{sec: training type}, embedding the training module may alter gradient descent trajectories, potentially causing prediction fluctuations.

The experimental results for +ECF$_{\mathcal{S}}$ show that our framework reliably enhances model performance without participating in training, as seen in examples like UNET achieving a $26.4\%$ error reduction and CNO showing a $23.7\%$ improvement on AC-DW problems. The stable improvement validates Theorem~\ref{theorem 2}. Specifically, +ECF$_{\mathcal{S}}$ reduces errors in the prediction's zero-frequency component while preserving other frequencies, leading to reliable performance gains.

Furthermore, our findings reveal that in most cases, the +ECF$_{\mathcal{S}}$ framework yields prediction errors either larger than or comparable to those of the +ECF$_{\mathcal{I}}$ framework. Notably, the best-performing model on each dataset consistently comes from a baseline integrated with the +ECF$_{\mathcal{I}}$ framework. This demonstrates that incorporating our framework during model training typically provides beneficial guidance to the model.
\begin{table}[]
\caption{Computational time cost for training and testing on the AC-DW dataset.}
\label{tab: time}
\renewcommand{\arraystretch}{1.1} 
\setlength{\tabcolsep}{15pt}
\begin{tabular}{l|cc}
\hline
Model & Training Time (s) & Test Time (s) \\
\hline
FNO & $1.92\text{E+}03$ & $4.11\text{E+}00$ \\
+ECF$_{\mathcal{I}}$ & $1.95\text{E+}03$ & $4.43\text{E+}00$ \\
+ECF$_{\mathcal{S}}$ & $1.92\text{E+}03$ & $4.15\text{E+}00$ \\
\hline
UNO & $3.57\text{E+}03$ & $7.81\text{E+}00$ \\
+ECF$_{\mathcal{I}}$ & $3.61\text{E+}03$ & $7.96\text{E+}00$ \\
+ECF$_{\mathcal{S}}$ & $3.57\text{E+}03$ & $8.17\text{E+}00$ \\
\hline
CNO & $5.82\text{E+}03$ & $1.40\text{E+}01$ \\
+ECF$_{\mathcal{I}}$ & $5.83\text{E+}03$ & $1.43\text{E+}01$ \\
+ECF$_{\mathcal{S}}$ & $5.82\text{E+}03$ & $1.44\text{E+}01$ \\
\hline
UNet & $3.23\text{E+}03$ & $9.49\text{E+}00$ \\
+ECF$_{\mathcal{I}}$ & $3.24\text{E+}03$ & $9.82\text{E+}00$ \\
+ECF$_{\mathcal{S}}$ & $3.23\text{E+}03$ & $9.57\text{E+}00$ \\
\hline
Transolver & $3.78\text{E+}04$ & $7.13\text{E+}01$ \\
+ECF$_{\mathcal{I}}$ & $3.78\text{E+}04$ & $7.16\text{E+}01$ \\
+ECF$_{\mathcal{S}}$ & $3.78\text{E+}04$ & $7.13\text{E+}01$ \\
\hline
\end{tabular}
\vspace{-1.5 mm}
\end{table}
\subsection{Computational Efficiency Analysis}
To demonstrate the impact of our framework on model training and inference time costs, Table~\ref{tab: time} compares training (500 samples/1k epochs) and testing (100 samples) time costs before/after framework integration, with other parameters detailed in Appendix~\ref{app: training details}.

Table~\ref{tab: time} presents the training and test time costs of various models on the AC-DW dataset.
Since the +ECF$_{\mathcal{S}}$ framework first independently trains the neural operator before applying the correction framework during test, its training time remains identical to the baseline.
Our results show that +ECF$_{\mathcal{I}}$ slightly increases both training and test time costs across most models, while +ECF$_{\mathcal{S}}$ only marginally increases test time. However, these increments are negligible. For instance, on the UNO model, +ECF$_{\mathcal{I}}$ only introduces a $1.02\%$ training time increase and a $1.92\%$ test time increase. In summary, although our frameworks introduce minor time cost increases, these are extremely small relative to the models' overall training and inference time.

\subsection{Conservation Error}

\begin{table}[h]
\centering
\caption{Comparison of conservation quantity errors between different models and their corresponding +ECF frameworks across various physical tasks, where the error is measured by the average relative error of conservation quantities.}
\label{tab: all conservation error}
\small 
\renewcommand{\arraystretch}{1} 
\setlength{\tabcolsep}{2.5pt} 
\begin{tabular}{l|cccccc}
\toprule
 & AC-DW & AC-FH & Heat & Water & Diff & CD \\
\midrule
FNO & 2.00E-01 & 1.76E-01 & 4.77E-02 & 6.65E-03 & 5.69E-03 & 9.72E-02 \\
+ECF$_{\mathcal{I}}$ & 2.00E-06 & 1.52E-06 & 2.42E-06 & 2.45E-06 & 1.25E-06 & 2.92E-06 \\
+ECF$_{\mathcal{S}}$ & 2.03E-06 & 1.57E-06 & 2.45E-06 & 2.63E-06 & 1.24E-06 & 2.83E-06 \\
\midrule
UNO & 2.41E-01 & 7.19E-02 & 1.69E-01 & 1.29E-02 & 6.46E-03 & 1.81E-01 \\
+ECF$_{\mathcal{I}}$ & 1.97E-06 & 1.56E-06 & 2.44E-06 & 2.10E-06 & 1.26E-06 & 2.93E-06 \\
+ECF$_{\mathcal{S}}$ & 1.97E-06 & 1.54E-06 & 2.33E-06 & 1.68E-06 & 1.22E-06 & 2.70E-06 \\
\midrule
CNO & 9.59E-01 & 1.16E+00 & 1.15E+00 & 1.25E-02 & 1.31E-02 & 1.35E+00 \\
+ECF$_{\mathcal{I}}$ & 1.47E-07 & 2.90E-07 & 1.60E-07 & 3.35E-09 & 1.97E-07 & 2.07E-07 \\
+ECF$_{\mathcal{S}}$ & 1.80E-07 & 2.38E-07 & 1.72E-07 & 2.95E-08 & 7.76E-08 & 2.00E-07 \\
\midrule
UNet & 1.42E+00 & 2.36E+00 & 1.21E+00 & 8.57E-03 & 1.33E-02 & 1.51E+00 \\
+ECF$_{\mathcal{I}}$ & 1.64E-07 & 3.27E-07 & 1.40E-07 & 1.13E-08 & 2.07E-07 & 2.29E-07 \\
+ECF$_{\mathcal{S}}$ & 1.88E-07 & 2.35E-07 & 1.86E-07 & 5.06E-08 & 7.66E-08 & 2.32E-07 \\
\midrule
Transolver & 4.60E-01 & 2.20E-01 & 1.29E+00 & 3.42E-02 & 5.81E-03 & 3.03E-01 \\
+ECF$_{\mathcal{I}}$ & 2.59E-06 & 1.90E-06 & 1.79E-06 & 1.65E-06 & 1.25E-06 & 3.29E-06 \\
+ECF$_{\mathcal{S}}$ & 1.90E-06 & 1.44E-06 & 1.80E-06 & 2.43E-06 & 1.20E-06 & 2.87E-06 \\
\bottomrule
\end{tabular}
\end{table}

\begin{figure*}[htbp]
    \centering
    \subfigure{\includegraphics[width=0.3\linewidth]{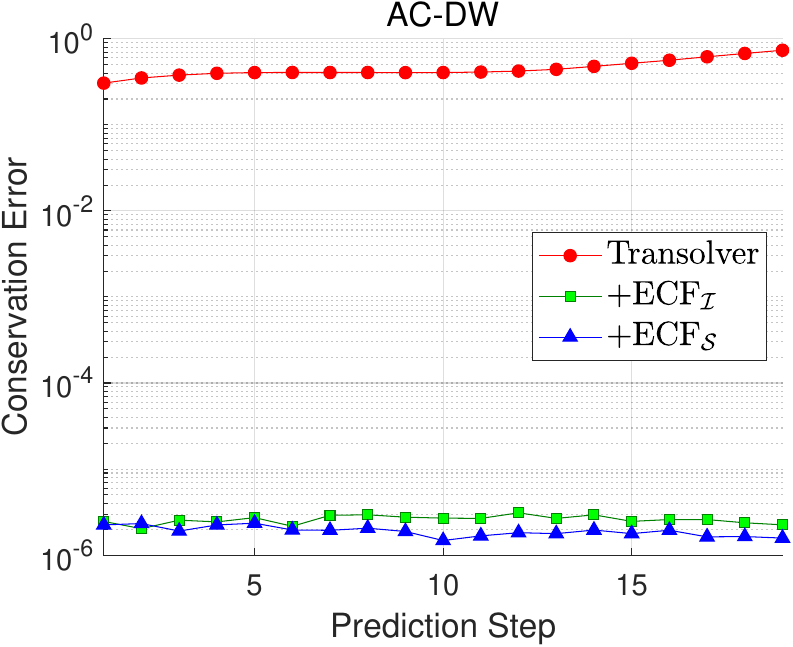}\label{fig:p1}}  
    \subfigure{\includegraphics[width=0.3\linewidth]{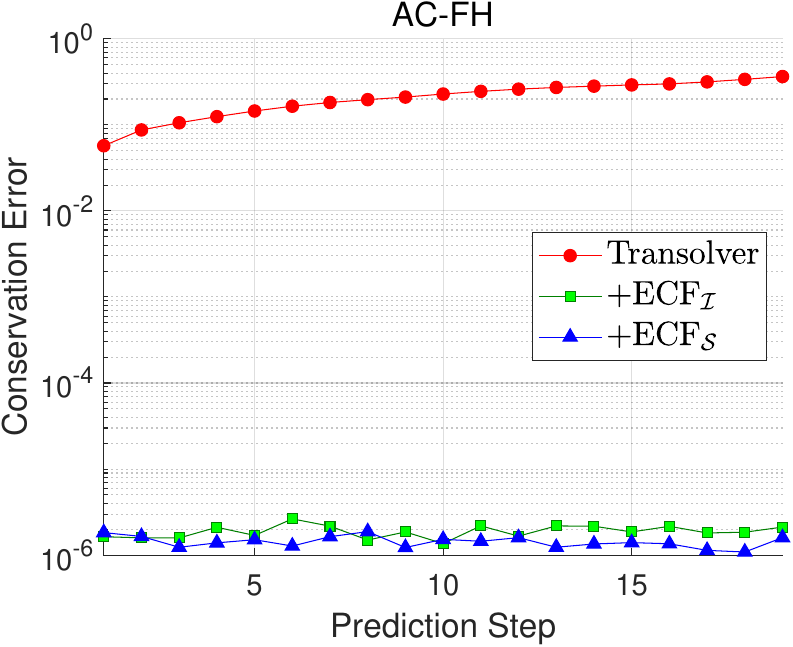}\label{fig:p2}} 
    \subfigure{\includegraphics[width=0.3\linewidth]{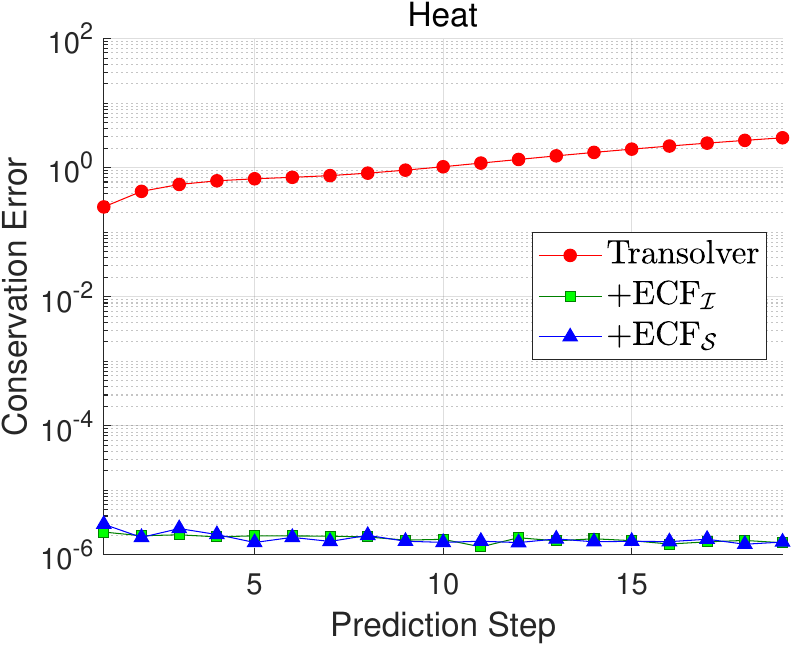}\label{fig:p3}}

    \vspace{-2mm}
    \subfigure{\includegraphics[width=0.3\linewidth]{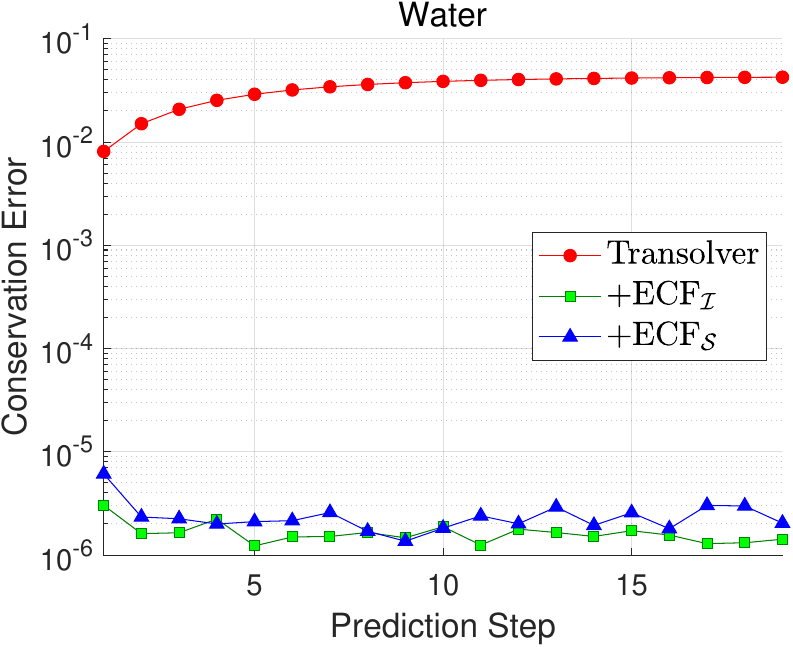}\label{fig:p4}} 
    \subfigure{\includegraphics[width=0.3\linewidth]{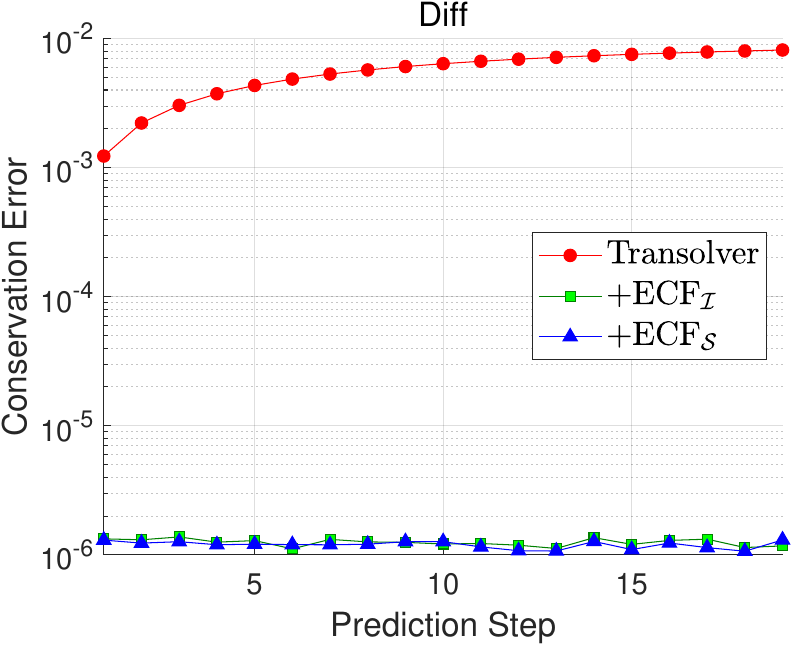}\label{fig:p5}}  
    \subfigure{\includegraphics[width=0.3\linewidth]{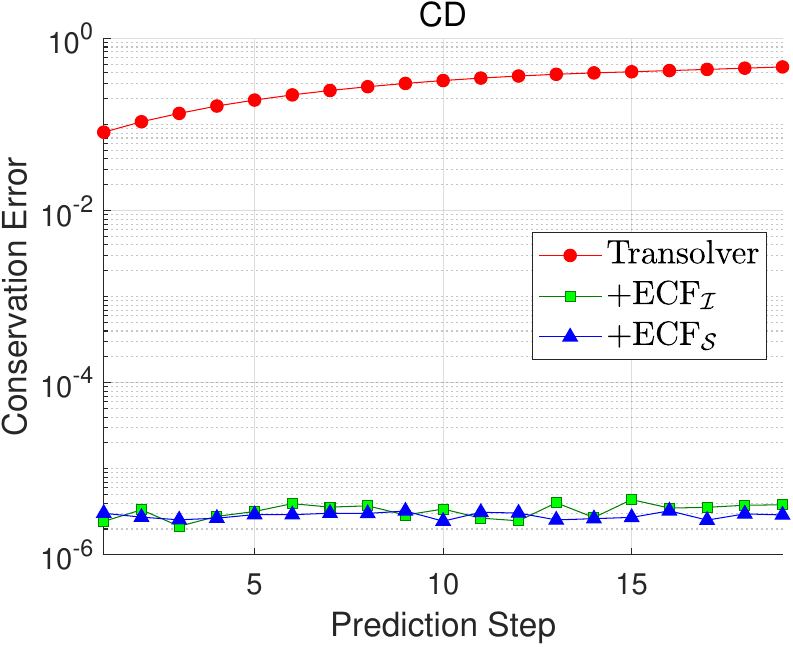}\label{fig:p6}}

    \caption{Time-evolution of relative conservation errors for baseline Transolver and its corresponding +ECF frameworks (+ECF$_{\mathcal{I}}$/+ECF$_{\mathcal{S}}$) in all datasets}
    \label{fig: different datset's Conservation Error}
\end{figure*}

In Figure~\ref{fig: different datset's Conservation Error}, we present the relative conservation errors of the predictions obtained by the original Transolver and its variants with +ECF$_{\mathcal{I}}$ and +ECF$_{\mathcal{S}}$ architectures across six datasets. The relative conservation error is defined as:

\begin{equation}
\begin{aligned}
Error(t) = \frac{\left|\int_{\Omega}u_{pred}(x,t)dx - \int_{\Omega}u_{true}(x,t)dx\right|}{\left|\int_{\Omega}u_{true}(x,t)dx\right|}
\end{aligned}
\end{equation}

Table~\ref{tab: all conservation error} details the conservation errors for each model. Results show that while all baseline models exhibit significant and temporally accumulating conservation errors, our framework effectively eliminates such errors. The relative conservation errors even exceed $100\%$ in some datasets - for instance, the UNet model shows a relative error of 2.36 (i.e., $236\%$) on the AC-FH dataset.
After integrating our framework, both the +ECF$_{\mathcal{I}}$ and +ECF$_{\mathcal{S}}$ errors are consistently maintained at or below the 1E-6 magnitude level. It should be noted that the non-zero errors in the conservation quantity arise because data stored in float32 floating-point format introduces inherent machine errors on the order of approximately 1E-7, while the values of the conserved quantity in the data itself are at the 1E-1 to 1E+0 magnitude level.
This factor results in a theoretical upper bound of 1E-6 magnitude for the conservation error in the framework's final output, which is fully consistent with our experimental results. This indicates that our proposed framework effectively enables the model to capture conservation laws and strictly enforces them, thereby resolving the issue of accumulating conservation errors in long-term predictions.

Furthermore, by examining Tables~\ref{tab: main experiment} and~\ref{tab: all conservation error}, it can be observed that when the conservation error in the model's predictions is relatively small, the performance improvement brought by our framework is quite limited. For instance, on the Diff dataset, the maximum relative conservation error of the model is only 1.33E-02 ($1.33\%$). In this case, when integrated into the model, our framework yields performance improvements not exceeding $5\%$ for all models except the FNO model.

%% file: component/7_conclusion.tex
\section{Conclusion}
We propose the \textbf{Exterior-Embedded Conservation Framework (ECF)} that enforces physical conservation laws in neural operators via frequency-domain correction. Featuring two training paradigms \textbf{Integrated Mode (ECF$_{\mathcal{I}}$)} and \textbf{Staged Mode (ECF$_{\mathcal{S}}$)} adapts to diverse optimizations without architectural changes. We establish the first quantitative relationship between conservation errors and RMSE while preserving model expressiveness. Experimental results demonstrate ECF's superior performance on conservation law-constrained PDEs, offering a rigorous yet practical solution for physics-informed machine learning.

%% file: component/8_appendix.tex
\newpage
\section{Appendix}

\subsection{Experient Details}\label{app: Experient Details}
\subsubsection{Dataset Details}\label{app: dataset}
Here, we list the PDEs of the datasets we used.

\begin{itemize}
    \item AC-DW: We consider the two-dimensional Allen-Cahn equation with the derivative of a double-well potential as the nonlinear coefficient function, which takes the form:
    \begin{equation}
    \begin{aligned} \
    \partial_t u = \epsilon \nabla^2 u + u -& u^3 - \frac{1}{|\Omega|}\int_\Omega (u - u^3)dx,\ & x\in\Omega,\ t\in[0,T]\\
    u(\bm x,0) &= u_0(\bm x),\ & x\in\Omega.
    \end{aligned}
    \end{equation}

    The dataset is generated by numerically solving the conservation equation on the unit square domain $\Omega = [0,1]^2$ discretized with a $128 \times 128$ uniform spatial grid, using an interface parameter $\epsilon = 0.01$ and periodic boundary conditions in both spatial dimensions. The temporal evolution is computed through $N_t = 1000$ time steps up to $T = 0.1$, with solutions sampled at $20$ equidistant time intervals. Initial conditions are constructed as random linear combinations of 20th-order two-dimensional Chebyshev polynomials, $u_0(x,y) = \sum_{i,j=0}^{19} c_{ij}T_i(x)T_j(y)$, where $c_{ij} \sim \mathcal{U}[-1,1]$.

    \item AC-FH: We consider the two-dimensional conserved Allen-Cahn equation with Flory-Huggins logarithmic potential, which takes the form:
    \begin{equation}
    \begin{aligned}
    &\partial_t u = \epsilon \nabla^2 u + \frac{\theta}{2}\ln\left(\frac{1+u}{1-u}\right) - \theta_c u  \\
    &- \frac{1}{|\Omega|}\int_\Omega \frac{\theta}{2}\ln\left(\frac{1+u}{1-u}\right) - \theta_c udx,\\ 
    &\bm{x}\in\Omega,\ t\in[0,T]\\
    &u(\bm{x},0) = u_0(\bm{x}),  \bm{x}\in\Omega.
    \end{aligned}
    \end{equation}
    
    The dataset is generated by numerically solving the conservation equation on the unit square domain $\Omega = [0,1]^2$ discretized with a $64 \times 64$ uniform spatial grid. The system parameters are set to $\epsilon = 0.01$, $\theta = 0.8$, and $\theta_c = 1.6$, with periodic boundary conditions enforced in both spatial dimensions. Temporal evolution is computed through $N_t = 1000$ time steps up to $T = 0.1$, with solutions sampled at $20$ equidistant time intervals. The initial condition $u_0(\bm{x})$ is constructed via random linear combinations of 20th-order two-dimensional Chebyshev polynomials:$u_0(\bm{x}) = \sum_{i,j=0}^{19} c_{ij}T_i(x)T_j(y), \quad c_{ij} \sim \mathcal{U}[-1,1]$, where $\bm{x} = (x,y) \in \Omega$.

    \item Heat: We consider the two-dimensional Adiabatic Process, which takes the form:
    \begin{equation}
    \begin{aligned}
        \partial_t u = D\nabla^2 u, \bm{x}\in\Omega,\ t\in[0,T]\\
        \nabla u\cdot \bm{n} = 0 \bm{x}\in \partial\Omega.
    \end{aligned}
    \end{equation}
    The dataset is generated by numerically solving the conservation equation on the unit square domain $\Omega = [0,1]^2$ discretized with a $128 \times 128$ uniform spatial grid, using a diffusion coefficient $D=0.01$ and adiabatic boundary conditions in both spatial dimensions.
    Temporal evolution is computed through $N_t = 1000$ time steps up to $T = 1$, with solutions sampled at $20$ equidistant time intervals. The initial condition $u_0(\bm{x})$ is constructed via random linear combinations of 20th-order two-dimensional Chebyshev polynomials:$u_0(\bm{x}) = \sum_{i,j=0}^{19} c_{ij}T_i(x)T_j(y), \quad c_{ij} \sim \mathcal{U}[-1,1]$, where $\bm{x} = (x,y) \in \Omega$.
    
    \item Water: We consider the two-dimensional Shallow Water Equation, which takes the form:
    \begin{equation}
        \begin{aligned}
            \partial_{t} h + \nabla \cdot (h \bm{u}) &= 0,\\
            \partial_{t}(h \bm{u}) + \nabla \cdot \left(\frac{1}{2} h \bm{u}^{2} + \frac{1}{2} g_{r} h^{2}\right) &= -g_{r} h \nabla b.
        \end{aligned}
    \end{equation}
    We need to predict water depth $h(\bm x, t)$. We employ the dataset "2D\_rdb\_NA\_NA.h5" from PDEbench as our benchmark dataset. We sample the solutions at 20 equidistant time intervals as the experimental dataset.

    \item Diff: We consider the two-dimensional Diffusion Equation, which takes the form:
    \begin{equation}
    \begin{aligned}
        \partial_t u &= D\nabla^2 u, &\bm{x}\in\Omega,\ t\in[0,T]\\
        u(\bm{x},0) &= u_0(\bm{x}),  &\bm{x}\in\Omega.
    \end{aligned}
    \end{equation}
    The dataset is generated by numerically solving the conservation equation on the unit square domain $\Omega = [0,1]^2$ discretized with a $100 \times 100$ uniform spatial grid, using a diffusion coefficient $D=0.01$ and periodic boundary conditions in both spatial dimensions.
    Temporal evolution is computed through $N_t = 1000$ time steps up to $T = 1$, with solutions sampled at $20$ equidistant time intervals. The initial state $u_0(\bm{x})$ is generated via a Gaussian random field $GRF(\tau = 5, \alpha = 2)$~\cite{dong2024accelerating}.

    \item CD: We consider the two-dimensional Convection-Diffusion Equation, which takes the form:
    \begin{equation}
        \begin{aligned}
            \frac{\partial \phi}{\partial t} + u \frac{\partial \phi}{\partial x} + v \frac{\partial \phi}{\partial y} = D \left( \frac{\partial^2 \phi}{\partial x^2} + \frac{\partial^2 \phi}{\partial y^2} \right)
        \end{aligned}
    \end{equation}
    We solve for the diffusion process of the concentration field $\phi(\bm x,t)$ within a uniform flow field with constant velocity. The dataset is generated by numerically solving the conservation equation on the unit square domain $\Omega = [0,1]^2$ discretized with a $128 \times 128$ uniform spatial grid.
    The flow field maintains constant velocities of $u=1.0$ in the $x$-direction and $v=0.5$ in the $y$-direction, with a diffusion coefficient $D=0.01$, while periodic boundary conditions are imposed along both spatial dimensions. Temporal evolution is computed through $N_t = 1000$ time steps up to $T = 0.1$, with solutions sampled at $20$ equidistant time intervals. Initial conditions are constructed as random linear combinations of 20th-order two-dimensional Chebyshev polynomials, $u_0(x,y) = \sum_{i,j=0}^{19} c_{ij}T_i(x)T_j(y)$, where $c_{ij} \sim \mathcal{U}[-1,1]$.
\end{itemize}
\subsubsection{Model Details}\label{app: model details}
The models used in our paper are all implemented based on their official source code provided on GitHub. Table~\ref{tab: model details} summarizes the key architectural hyperparameters of these models.
\begin{table}[h]
\centering
\caption{Model Parameters (continued from previous page)}
\begin{tabular}{lll}
\toprule
\textbf{Model Name} & \textbf{Parameter} & \textbf{Value(s)} \\
\midrule
 & Output Channels (C\_out) & 1 \\
\multirow{4}{*}{FNO} 
& Modes(modes1,modes2)) & 12,12 \\
 & Width & 20\\
 & Input Channels(C\_in) & 1\\
 & Output Channels(C\_out)  & 1 \\
\midrule
\multirow{5}{*}{CNO} 
& Input Channels (in\_dim) & 1\\
 & Num Layers (N\_layers) & 2 \\
 & Channel Multiplier & 16\\
 & Latent Lift Proj Dim & 8\\
 & Activation & 'cno\_lrelu' \\
\midrule
\multirow{4}{*}{UNO} 
& Input Width & 5\\
& Width & 32\\
& Domain Padding & 0 \\
& Scaling Factor & 3/4 \\
\midrule
\multirow{4}{*}{U-Net}  
 & Input/Output Channels & 1 / 1 \\
 & Kernel Size & 3 \\
 & Base Channels & 64 \\
 & Downsample Levels & 4 \\
\midrule
\multirow{5}{*}{Transolver}
 & Space\_dim & 1 \\
 & Num Layers & 8 \\
 & Num Head & 8 \\
 & Num Hidden & 256 \\
 & Unified Pos & 1 \\
\bottomrule
\label{tab: model details}
\end{tabular}
\end{table}

\subsubsection{Training Details}\label{app: training details}
Our experimental setup employs the AdamW optimizer with a fixed learning rate of $1×10^{-3}$ across all models. The training protocol utilizes a dataset of $500$ samples, complemented by $100$-sample validation and test sets, with a consistent batch size of $5$ throughout the $1,000$-epoch training process. To ensure model selection robustness, we perform evaluations at $50$-epoch intervals and retain the best-performing checkpoint. All results represent the average of five independent runs with fixed random seeds $(0,1,2,3,4)$ for reproducibility. The computational infrastructure consists of NVIDIA 3090 GPU servers, providing the necessary resources for this intensive training regimen. Unless otherwise specified, we adopt the mean absolute error (MAE) as the loss function during training, while using the root mean square error (RMSE) as the evaluation metric for prediction quality:
\begin{equation}
\begin{aligned}
    \text{MSE}(y,\hat{y}) &= \frac{1}{n}\sum_{i=1}^{n}(y_i - \hat{y}_i)^2\\
    \text{RMSE}(y,\hat{y}) &= \sqrt{\frac{1}{n}\sum_{i=1}^{n}(y_i - \hat{y}_i)^2}
\end{aligned}
\end{equation}
During the training phase, we randomly select data pairs from two consecutive time steps within a 20-step window for each prediction task. For validation and testing, the model sequentially predicts all subsequent time steps using only the initial time step's data as input. 

Notably, the +ECF$_{\mathcal{S}}$ framework employs a unique embedding strategy: its correction mechanism remains inactive during both training and validation phases, activating exclusively during testing. The implementation involves applying our proposed correction scheme to adjust conservation errors at each time step after the model completes predictions across the entire temporal sequence.

\subsection{Theorem Proof}\label{app: theorem proof}
\setcounter{theorem}{0}
\begin{theorem}\label{eq: theorem gauss}
The time derivative of the integrated quantity satisfies:
\begin{equation}
    \begin{aligned}
        \frac{d\bm{E}}{dt} = -\oint_{\partial\Omega}\bm{F}(\bm{u})\cdot\bm{n}dS + \int_{\Omega}\bm{S}d\bm{x},
    \end{aligned}
\end{equation}
where $\partial \Omega$ denotes the boundary of the domain $\Omega$ and $\bm{n}$ is the outward unit normal vector on $\partial\Omega$.
\end{theorem}

\begin{proof}
    Clearly, we have
    \begin{equation}
        \begin{aligned}
            \frac{d\bm{E}}{dt} = \frac{d}{dt} \int_{\Omega}\bm{u}(\bm{x}, t)d\bm{x} = \int_{\Omega} \frac{\partial \bm{u}}{\partial t} d\bm{x}.
        \end{aligned}
    \end{equation}
    Since 
    \begin{equation}
        \begin{aligned}
            \frac{\partial \bm{u}}{\partial t} + \nabla \cdot \bm{F}(\bm{u}) = \bm{S}(\bm{u},\bm{x},t),
        \end{aligned}
    \end{equation}
    we can see that 
    \begin{equation}
        \begin{aligned}
             \int_{\Omega} \frac{\partial \bm{u}}{\partial t} d\bm{x} 
             & = \int_{\Omega} \bm{S}(\bm{u},\bm{x},t) - \nabla \cdot \bm{F}(\bm{u})\ d\bm{x} \\
             & = - \int_{\Omega} \nabla \cdot \bm{F}(\bm{u}) d\bm{x}+ \int_{\Omega}\bm{S}d\bm{x}.
        \end{aligned}
    \end{equation}
    Then, by Gauss's divergence theorem, we have 
    \begin{equation}
        \begin{aligned}
            \int_{\Omega} \nabla \cdot \bm{F}(\bm{u}) d\bm{x} = \oint_{\partial\Omega}\bm{F}(\bm{u})\cdot\bm{n}dS.
        \end{aligned}
    \end{equation}
    All together, we can conclude that
    \begin{equation}
        \begin{aligned}
            \frac{d\bm{E}}{dt} 
            & = \int_{\Omega} \frac{\partial \bm{u}}{\partial t} d\bm{x} \\
            & = - \int_{\Omega} \nabla \cdot \bm{F}(\bm{u}) d\bm{x}+ \int_{\Omega}\bm{S}d\bm{x} \\
            & = -\oint_{\partial\Omega}\bm{F}(\bm{u})\cdot\bm{n}dS + \int_{\Omega}\bm{S}d\bm{x}.
        \end{aligned}
    \end{equation}
\end{proof}
\begin{theorem}\label{theorem: frequent and error}
    The error between the predicted result $\bm{v}^{t+1}$ and the ground truth $\hat{\bm{v}}^{t+1}$ can be expressed in terms of their Fourier coefficients. Specifically:
    \begin{equation}
        \begin{aligned}
            &\|\bm{v}^{t+1} - \hat{\bm{v}}^{t+1}\|_{L^2(\Omega)}^2\\ 
            &= L^m \sum_{\bm{n}\in\mathbb{Z}^m} |\bm{c}_{\bm{n}}(t+1) - \hat{\bm{c}}_{\bm{n}}(t+1)|^2 \\
            &= L^m \Big[ |\bm{c}_{\mathbf{0}}(t+1) - \hat{\bm{c}}_{\mathbf{0}}(t+1)|^2 \\ 
            &\ \quad + \sum_{\bm{n}\in\mathbb{Z}^m\setminus\{\mathbf{0}\}} |\bm{c}_{\bm{n}}(t+1) - \hat{\bm{c}}_{\bm{n}}(t+1)|^2 \Big],
        \end{aligned}
    \end{equation}
    where $\|\cdot\|_{L^2(\Omega)}$ represents the $L^2$-norm on domain $\Omega$, defined as $\| f\|_{L^2(\Omega)} = \left(\int_{\Omega}|f(\bm{x})|^2d\bm{x}\right)^{\frac{1}{2}}$. 
\end{theorem}

\begin{proof}
    By the Fourier decomposition, we can see that \begin{equation}
        \begin{aligned}
            & \| \bm{v}^{t+1} - \hat{\bm{v}}^{t+1}\|_{L^2(\Omega)}^2 \\
            & = \| \sum_{\bm{n}\in\mathbb{Z}^m} \bm {c}_{\bm{n}}(t+1) e^{i\frac{2\pi}{L}\bm{n}\cdot\bm{x}} - \sum_{\bm{n}\in\mathbb{Z}^m} \hat{\bm {c}}_{\bm{n}}(t+1) e^{i\frac{2\pi}{L}\bm{n}\cdot\bm{x}} \|_{L^2(\Omega)}^2 \\
            & = \| \sum_{\bm{n}\in\mathbb{Z}^m} \big( \bm{c}_{\bm{n}}(t+1) - \hat{\bm {c}}_{\bm{n}}(t+1) \big) e^{i\frac{2\pi}{L}\bm{n}\cdot\bm{x}} \|_{L^2(\Omega)}^2 \\
            & = \sum_{\bm{n}\in\mathbb{Z}^m} \| e^{i\frac{2\pi}{L}\bm{n}\cdot\bm{x}} \|_{L^2(\Omega)}^2 | \bm{c}_{\bm{n}}(t+1) - \hat{\bm {c}}_{\bm{n}}(t+1) |^{2}.
        \end{aligned}
    \end{equation}
    The last equality above holds because of the orthogonality of the Fourier basis $e^{i\frac{2\pi}{L}\bm{n}\cdot\bm{x}}$.
    Noting the fact that
    \begin{equation}
        \begin{aligned}
            \| e^{i\frac{2\pi}{L}\bm{n}\cdot\bm{x}} \|_{L^2(\Omega)}^2 = L^{m},
        \end{aligned}
    \end{equation}
    we have 
    \begin{equation}
        \begin{aligned}
            & \sum_{\bm{n}\in\mathbb{Z}^m} \| e^{i\frac{2\pi}{L}\bm{n}\cdot\bm{x}} \|_{L^2(\Omega)}^2 | \bm{c}_{\bm{n}}(t+1) - \hat{\bm {c}}_{\bm{n}}(t+1) |^2 \\
            & = L^m \sum_{\bm{n}\in\mathbb{Z}^m} |\bm{c}_{\bm{n}}(t+1) - \hat{\bm{c}}_{\bm{n}}(t+1)|^2 \\
            & = L^m \Big[ |\bm{c}_{\mathbf{0}}(t+1) - \hat{\bm{c}}_{\mathbf{0}}(t+1)|^2 \\
            &\ \quad + \sum_{\bm{n}\in\mathbb{Z}^m\setminus\{\mathbf{0}\}} |\bm{c}_{\bm{n}}(t+1) - \hat{\bm{c}}_{\bm{n}}(t+1)|^2 \Big].
        \end{aligned}
    \end{equation}
    All together, we can conclude that
    \begin{equation}
        \begin{aligned}
            & \| \bm{v}^{t+1} - \hat{\bm{v}}^{t+1}\|_{L^2(\Omega)}^2 \\
            & = \sum_{\bm{n}\in\mathbb{Z}^m} \| e^{i\frac{2\pi}{L}\bm{n}\cdot\bm{x}} \|_{L^2(\Omega)}^2 | \bm{c}_{\bm{n}}(t+1) - \hat{\bm {c}}_{\bm{n}}(t+1) |^{2} \\
            & = L^m \sum_{\bm{n}\in\mathbb{Z}^m} |\bm{c}_{\bm{n}}(t+1) - \hat{\bm{c}}_{\bm{n}}(t+1)|^2 \\
            & = L^m \Big[ |\bm{c}_{\mathbf{0}}(t+1) - \hat{\bm{c}}_{\mathbf{0}}(t+1)|^2 \\
            &\ \quad + \sum_{\bm{n}\in\mathbb{Z}^m\setminus\{\mathbf{0}\}} |\bm{c}_{\bm{n}}(t+1) - \hat{\bm{c}}_{\bm{n}}(t+1)|^2 \Big].
        \end{aligned}
    \end{equation}
\end{proof}

\begin{theorem}\label{theorem: 2}
Let $\hat{\bm{v}}^{t+1}$, $\bm{v}^{t+1}$, and $\bar{\bm{v}}^{t+1}$ denote the ground truth at time $t+1$, the solution function predicted by the neural operator, and the corrected solution function obtained via Equation~\ref{eq:Corrected conserved quantities}, respectively. Then:
\begin{equation}
\| \bar{\bm{v}}^{t+1} - \hat{\bm{v}}^{t+1} \|_{L^2(\Omega)} \leq \| \bm{v}^{t+1} - \hat{\bm{v}}^{t+1} \|_{L^2(\Omega)}.
\end{equation}
The equality holds if and only if:
\begin{equation}
    \int_\Omega \bm{v}^{t+1}d\bm{x} = \int_\Omega \hat{\bm{v}}^{t+1} d\bm{x}.
\end{equation}
\end{theorem}

\begin{proof}
    By the conclusion of Theorem \ref{thorem: frequent and error}, we can decompose the error as follows:
    \begin{equation}
        \begin{aligned}
            & \| \bar{\bm{v}}^{t+1} - \hat{\bm{v}}^{t+1} \|^{2}_{L^2(\Omega)} \\
            & = L^m \Big[ |\hat{\bm{c}}_{\mathbf{0}}(t+1) - \hat{\bm{c}}_{\mathbf{0}}(t+1)|^2 \\
            &\ \quad + \sum_{\bm{n}\in\mathbb{Z}^m\setminus\{\mathbf{0}\}} |\bm{c}_{\bm{n}}(t+1) - \hat{\bm{c}}_{\bm{n}}(t+1)|^2 \Big] \\
            & = L^{m}\sum_{\bm{n}\in\mathbb{Z}^m\setminus\{\mathbf{0}\}} |\bm{c}_{\bm{n}}(t+1) - \hat{\bm{c}}_{\bm{n}}(t+1)|^2,
        \end{aligned}
    \end{equation}
    and 
    \begin{equation}
        \begin{aligned}
            & \| \bm{v}^{t+1} - \hat{\bm{v}}^{t+1} \|^{2}_{L^2(\Omega)} \\
            & = L^m \Big[ |\bm{c}_{\mathbf{0}}(t+1) - \hat{\bm{c}}_{\mathbf{0}}(t+1)|^2 \\ 
            &\ \quad + \sum_{\bm{n}\in\mathbb{Z}^m\setminus\{\mathbf{0}\}} |\bm{c}_{\bm{n}}(t+1) - \hat{\bm{c}}_{\bm{n}}(t+1)|^2 \Big].
        \end{aligned}
    \end{equation}
    Thus, we have 
    \begin{equation}
        \begin{aligned}
            & \| \bm{v}^{t+1} - \hat{\bm{v}}^{t+1} \|^{2}_{L^2(\Omega)} - \| \bar{\bm{v}}^{t+1} - \hat{\bm{v}}^{t+1} \|^{2}_{L^2(\Omega)} \\ 
            & = L^{m} |\bm{c}_{\mathbf{0}}(t+1) - \hat{\bm{c}}_{\mathbf{0}}(t+1)|^2 \geq 0,
        \end{aligned}
    \end{equation}
    which leads to
    \begin{equation}
        \begin{aligned}
            \| \bar{\bm{v}}^{t+1} - \hat{\bm{v}}^{t+1} \|_{L^2(\Omega)} \leq \| \bm{v}^{t+1} - \hat{\bm{v}}^{t+1} \|_{L^2(\Omega)}.
        \end{aligned}\label{eq: 51}
    \end{equation}
    The equality of Equation~\ref{eq: 51} holds if and only if
    \begin{equation}
        \begin{aligned}
            |\bm{c}_{\mathbf{0}}(t+1) - \hat{\bm{c}}_{\mathbf{0}}(t+1)|^2 = 0,
        \end{aligned}
    \end{equation}
    which means
    \begin{equation}
        \begin{aligned}
            \bm{c}_{\mathbf{0}}(t+1) = \hat{\bm{c}}_{\mathbf{0}}(t+1),
        \end{aligned}
    \end{equation}
    i.e.
    \begin{equation}
        \begin{aligned}
            \int_\Omega \bm{v}^{t+1}d\bm{x} = \int_\Omega \hat{\bm{v}}^{t+1} d\bm{x}.
        \end{aligned}
    \end{equation}
\end{proof}